\documentclass[preprint,10pt,authoryear]{elsarticle}

\usepackage{amssymb}
\usepackage{amsmath}
\usepackage{enumitem}
\usepackage{verbatim}
\usepackage{cleveref}
\usepackage{graphicx}
\usepackage{color}

\usepackage{amsthm}
\theoremstyle{definition}
\newtheorem{definition}{D}
\theoremstyle{assumption}
\newtheorem{assumption}{A}
\theoremstyle{lemma}
\newtheorem{lemma}{Lemma}
\theoremstyle{proposition}
\newtheorem{proposition}{Proposition}
\theoremstyle{theorem}
\newtheorem{theorem}{Theorem}
\theoremstyle{remark}
\newtheorem{remark}{Remark}

\journal{}

\begin{document}

\begin{frontmatter}



\title{Lyapunov-Based Stabilization and Control of the Stochastic Schr\"odinger Equation}


\author{Peyman Azodi, Alireza Khayatian, Peyman Setoodeh}

\address{}

\begin{abstract}
This paper presents a detailed Lyapunov-based theory to control and stabilize continuously-measured quantum systems, which are driven by Stochastic Schr\"odinger Equation (SSE). Initially, equivalent classes of states of a quantum system are defined and their properties are presented. With the help of equivalence classes of states, we are able to consider global phase invariance of quantum states in our mathematical analysis. As the second mathematical modelling tool, the conventional It\"o formula is further extended to non-differentiable complex functions. Based on this extended It\"o formula, a detailed stochastic stability theory is developed to stabilize the SSE. Main results of this proposed theory are sufficient conditions for stochastic stability and asymptotic stochastic stability of the SSE. Based on the main results, a solid mathematical framework is provided for controlling and analyzing quantum system under continuous measurement, which is the first step towards implementing weak continuous feedback control for quantum computing purposes.
\end{abstract}




\end{frontmatter}









\section{Introduction}
\par Building quantum computers have been one of the most notable research efforts in the past decades. The idea of quantum computing started with the famous proposal of Richard Feynman of \textit{Simulating physics with computers} \cite{feynman1982simulating}. Although the original idea of Feynman was to simulate a quantum mechanical system with another quantum system, it turned out that exploiting purely quantum mechanical features is significantly applicable to other engineering demands \cite{dowling2003quantum}. These engineering applications mainly include secure communication \cite{renner2008security}, fast and efficient data processing \cite{divincenzo1995quantum} in addition to simulating nano-scale systems. 
\par Computing with devices at nano scales demands building, manipulating, and efficiently measuring such tiny devices, which is a challenging task. There must be a worthwhile payoff for performing engineering tasks at this level despite these complications. The key element, which makes quantum systems unique, is \textit{quantum coherence}. This phenomenon lives at nano scales and disappears as the system approaches the \textit{classical limit}. This many-fold complexity of systems at nano scales makes it \textit{exponentially} hard to simulate them with classical computers and classical bits. That is why it is almost impossible to simulate and analytically calculate the energy levels of simple molecular structures with classical computers. On the other hand, it has been shown that some classical hard problems can be solved in \textit{polynomial time} on a hypothetical quantum computer \cite{shor1999polynomial}.
\par The lifesaver property of quantum coherence dies out (decoherence happens) in two main situations, namely uncontrolled interaction of the quantum system with environment, which is undesirable and almost unavoidable \cite{viola1999dynamical} and when the system is measured, which is unavoidable since at some point, the results must be acquired from the system. Therefore, measurement alters the state of the system and inevitably destroys its quantum features  \cite{busch2009no}. This seems to be a turning point in quantum computing. Fortunately, the introduction of weak measurement and weak values \cite{aharonov1988result,jozsa2007complex,brun2002simple} turned the page. By weakly measuring a quantum system, we can maintain the quantum coherence at the cost of acquiring noisier information about the system \cite{jacobs2014quantum}. Recently, implementing weak continuous measurement of qubits has attracted theoreticians and experimentalists in the area of quantum computing \cite{muhonen2018coherent,pfender2019high,ran2019error,shojaee2018optimal,gross2018qubit,vijay2012stabilizing}. Noisy information about the system is not the only complication induced by implementing weak continuous measurement;
dynamics of these quantum systems are also governed by Stochastic Schr\"odinger Equation (SSE) \cite{wiseman1996quantum}, which takes into account the back-action of weak continuous measurement in the form of additional stochastic terms to the ordinary Schr\"odinger Equation \cite{breuer2002theory,jacobs2006straightforward,wiseman1994quantum}. Controlling and analyzing the stochastic quantum system governed by SSE, requires more consideration and more sophisticated mathematical tools,  provided by stochastic and nonlinear control theory \cite{sontag2013mathematical,khalil1996noninear,sastry2013nonlinear,khasminskii2011stochastic, kushner1967stochastic}. In this paper, the aim is to provide a solid mathematical framework (through control theory) to analyze, stabilize, and manipulate quantum systems in the presence of weak continuous measurement, in order to ease the way of implementing this framework for quantum computing purposes.  
 \par Lyapunov theory was initially used for quantum systems analysis in \cite{grivopoulos2003lyapunov}. Following this idea, in early 2000's, many researchers started to use Lyapunov theory for controlling and stabilization of Schr\"odinger equation \cite{kuang2008lyapunov,mirrahimi2005lyapunov,wang2010analysis}. Lyapunov theory has been among the important theories to analyze quantum systems in recent years \cite{wang2014optimal,kuang2018lyapunov,ghaeminezhad2018preparation,cardona2020exponential,liu2019filter,jie2018stabilizing}. For instance, Lyapunov theory has been used to stabilize two-level quantum systems \cite{cardona2018exponential,qamar2019observer}. Despite its importance and experimental applications, stabilization of an open quantum system under continuous measurement has not received the attraction it deserves.
 \par In this paper, we use stochastic Lyapunov theory stability analysis of continuously -measured quantum systems. In quantum physics, global phase gauge invariance induces some difficulties to control purposes. This is due to the fact that some states may be distinguishable from a mathematical viewpoint, but represent the same quantum state. In order to overcome this difficulty, equivalent classes of quantum states are defined and some of their useful properties are studied. The conventional It\"o formalism is not capable of evaluating the stochastic increment for the proposed Lyapunov function. Hence, initially, the It\"o formalism is extended to non-analytic complex functions. Then, based on the extended It\"o formula, the stochastic increment of the Lyapunov function can be evaluated. The Lyapunov function considered in this paper, which is defined based on the well-known Hilbert-Schmidt distance \cite{prugovecki1982quantum}, is highly adaptable to the proposed extended It\"o formula. The main results of this paper are two stochastic stability theorems. The first theorem provides the conditions on the desired final state, the Hamiltonian, and the observable that guarantee stochastic stability of the quantum system. In the second part, it is shown that some additional conditions on the control Hamiltonians lead to stochastic asymptotic stability of the quantum system. In this paper, it is assumed that the control system is Equivalent State Controllable (ESC), i.e. it is always possible to find a set of control signals to make the transition between any two arbitrary equivalent classes of states in a finite time interval  \cite{d2007introduction}. The conditions proposed for asymptotic stochastic stability are reasonably more strict than Equivalent State Controllability (ESC) conditions, since in addition to steering the system, these conditions also guarantee stochastic stability.

\par This paper consists of the following sections: In section \ref{P11}, some preliminary background is presented and the problem is formulated. This section mainly introduces continuously-measured quantum systems and their dynamic behaviour. Also, in the problem formulation subsection, some assumptions are defined to be used in the next sections. In section \ref{p12}, the Lyapunov candidate is defined and the It\"o formula is extended to suite non-analytic complex functions. In section \ref{p13}, main results on the stochastic stability of quantum systems are developed. First, two definitions are presented, afterward, the main results are presented within some lemmas and theorems. In section \ref{p14}, the proposed theorem is validated via an illustrative example of a qubit system, and section \ref{p15} provides concluding remarks.

\section {Preliminaries and Problem Formulation}\label{P11}
\subsection{Quantum Dynamical Systems}
\par According to the postulates of quantum mechanics, our state of knowledge about a finite dimensional ($n$-level) quantum system can be described by a normalized vector $\left| \psi  \right\rangle  \in {S^{2n - 1}}$, where ${S^{2n - 1}}=\mathcal{H}$ is the unit hypersphere $\mathbb{C}^{n}$, together with the Euclidean inner product $\left\langle {.} \mathrel{\left | {\vphantom {. .}} \right. \kern-\nulldelimiterspace} {.} \right\rangle $. The time evolution of a closed quantum system is governed by the Schr\"odinger wave equation:
\begin{equation}
\begin{array}{l}
\frac{d}{{dt}}\left| {\psi (t)} \right\rangle  = \frac{{ - i}}{\hbar }H\left| {\psi (t)} \right\rangle , 
\left| {\psi (0)} \right\rangle  = \left| {{\psi _0}} \right\rangle 
\end{array}
\label{10}
\end{equation}
where $H$($iH \in su(n)$) is a bounded (and equivalently compact) self-adjoint operator, called system Hamiltonian. 
\par The relative phase invariance in quantum mechanics has led to some difficulties in modeling quantum control systems. By this invariance, all the states $e^{i\varrho}\left| \psi  \right\rangle$ with $\varrho\in\mathbb{R}$  correspond to the same quantum state. In order to overcome this difficulty, the following definition is adapted:
\begin {definition}
$\left[ {\left| {{\psi }} \right\rangle } \right] \doteq  \left\{ {{e^{i\varrho} }\left| {{\psi }} \right\rangle \left| {\varrho  \in \mathbb{R}} \right.} \right\}$ is the \textit{equivalence class} of $\left| {{\psi }} \right\rangle$ and two different quantum states $\left| {{\psi _1}} \right\rangle$ and $\left| {{\psi _2}} \right\rangle$ are said to be \textit {equivalent} if  $\left| {{\psi _1}} \right\rangle \in \left[ {\left| {{\psi_2 }} \right\rangle } \right]$.
\label{EQ}
\end {definition}
 The quotient space corresponding to this partitioning is infinite dimensional even if we consider a finite dimensional quantum system.\par The following lemma will be useful in the rest of this article:
\begin{lemma}
Consider two non-equivalent quantum states, i.e. $\left| {{\psi _1}} \right\rangle  \notin \left[ {\left| {{\psi _2}} \right\rangle } \right]$, there exists $\epsilon >0$ and a self adjoint operator $H'$(${\left\| {H'} \right\|_{HS}} = 1$, $H'\ne I$)\footnote{${\left\| {.} \right\|_{HS}}$ denotes the Hilbert-Schmidt norm} such that:
\begin{equation}
\left| {{\psi _1}} \right\rangle  ={e^{i\varepsilon H'}}  {\left| {{\psi _2}} \right\rangle }
\end{equation}
Also, neither of $\left| {{\psi _1}} \right\rangle$ and $\left| {{\psi _2}} \right\rangle$ is an eigenket of $H'$.
\label{lemma10}
\end{lemma}
\begin{proof}[Sketch of the proof]
The transitivity property of $SU(n)$ as a Lie transformation group and its Homeomorphism to $su(n)$ by the exponential map provide the main idea for the proof. Also, both of the following conditions:
\begin{enumerate}
\item {$H'= I$}
\item{Any one of $\left| {{\psi _1}} \right\rangle$ or $\left| {{\psi _2}} \right\rangle$ being an eigenket for $H'$}
\end{enumerate}
contradicts the non-equivalence property. 
\end{proof}

\subsection{Continuously measured quantum systems}
\label{COQS}
\par By the early postulates of quantum mechanics, the measurement process, projects the system into the eigenspace corresponding to the resulting eigenket of the observable. This viewpoint is known as projective measurement or also, strong Von Neumann measurement. By this postulate, the pre-measure quantum state is completely missed. Developing Positive Operator Valued Measure (POVM)  led to a more general model for quantum measurement, known as weak measurement \cite{brun2002simple,aharonov1988result}. By weak measurement, the pre-measure state, does not necessarily collapse to the eigenstates of the observable but deviates from its initial state. This deviation depends on the strength of the measurement and the distribution spread of the Gaussian measurement. Moreover, the obtained information is unsharp. This phenomena obeys the Busch's theorem: ``no information without disturbance"\cite {busch2009no}. 
\par A continuously measured quantum system is the one being continuously weakly measured. Let us denote the measurement observable by $X$, which is neccesarily self adjoint. In \cite{jacobs2006straightforward,jacobs2014quantum}, a Gaussian approach to weak measurement is followed. In this approach, the deviation of the premeasure state depends on the measurement strength and the energy shifts between the eigenkets of $X$ weighted by the projection of pre-measure state into the eigenbasis generated by $X$. In the case of continuous measurement, this approach leads to the following stochastic dynamic behaviour, known as Stochastic Schr\"odinger Equation:
\begin{equation}
d\left| \psi  \right\rangle  = \left( {( - k{{(X - \left\langle X \right\rangle )}^2})dt + \sqrt {2k} (X - \left\langle X \right\rangle )dW} \right)\left| \psi  \right\rangle,
\label{SSE1}
\end{equation}
where $k$ is the strength of the measurement, $dW$ is the standard Wiener process, and $\left\langle X \right\rangle$ is the expected value of the observable $X$. Also, it is shown that the contribution of the free evolution term in (\ref{10}) added together with (\ref{SSE1}), (i.e. the complete evolution) is given by: 
\begin{equation}
\begin{gathered}
  d\left| \psi  \right\rangle  = \left( {(\frac{{ - i}}{\hbar }H - k{{(X - \left\langle X \right\rangle )}^2})dt + \sqrt {2k} (X - \left\langle X \right\rangle )dW} \right)\left| \psi  \right\rangle.  \\ 
  \left| {\psi (0)} \right\rangle  = \left| {{\psi _0}} \right\rangle {\text{; a}}{\text{.s}}. \\ 
\end{gathered}
\label{SSE}
\end{equation}
\par This stochastic unitary evolution is a more general form of strong measurement. One may deduce from (\ref{SSE}) that in the case of great measurement strength, the quantum system is projected to the eigenkets of $X$ in a short time. For convenience let us use the following notation:
\begin{equation}
\hat f(\left| \psi  \right\rangle ) \doteq \left( {\frac{{ - i}}{\hbar }H - k{{(X - \left\langle X \right\rangle )}^2}} \right)\left| \psi  \right\rangle,\hat g(\left| \psi  \right\rangle ) \doteq \left( {\sqrt {2k} (X - \left\langle X \right\rangle )} \right)\left| \psi  \right\rangle,
\label{fg}
\end{equation}
where $\hat f$ and $\hat g$ are the drift and diffusion terms, respectively. Obviously, $\hat f$ and $\hat g$ admit the Lipschitz continuity and growth condition, thus, the existence and uniqueness of the strong solution of (\ref{SSE}) and the quantum trajectories are guaranteed.

\subsection {Problem formulation}
Dynamic behavior of an open quantum system with Hamiltonian $H$, which is continuously measured by $X$, was studied in section \ref{COQS}.  In order to address the control issue, consider the following assumptions:
\begin{assumption}
Consider the set of control signals $U \doteq \left\{ {{u_k}(t)\left| {k = 1,...,m;{u_k} \in {L^2}(\mathbb{R})}\right.}\right\}$. Each of the Lebesgue measurable control signals ${u_k}(t)$ is associated with a compact control Hamiltonian $H_k$ ($iH_k \in su(n)$) as its coefficient and they appear in the Hamiltonian in an affine manner:
\begin{equation}
{ H}(U) = {{ H}_0} + \sum\limits_k {{u_k}(t)} {{ H}_k}.
\label{HU}
\end{equation}
The free evolution Hamiltonian $H(0)=H_0$ describes the system in the absence of a control field.
\label{A10}
\end{assumption}
\par The previous assumption is not restrictive. In practice, the manipulation signals usually appear in an affine form (e.g. the effect of an interacting magnetic field on a spin system). Although the observable $X$ participates in the manipulation, in this article, we just address the role of control signals in the stabilization procedure. In the rest of this article, our aim is to analyze the stability properties of the quantum trajectories,  driven by the nonlinear stochastic differential equation (\ref{SSE}). The quantum trajectories start from the initial state $\left| {{\psi _0}} \right\rangle$ almost surely and the aim is to manipulate them to the equivalent class of the desired final state $\left[ {\left| {{\psi _f}} \right\rangle } \right]$. For convenience, the point spectrum (the compactness of $H_0$ and the separability of the underlying Hilbert space imply that the continuous and residual spectrum of $H_0$ are empty) of $H_0$ is denoted by $\sigma ({H_0}) \doteq \left\{ {{\lambda _i}} \right\}$ and the eigenspace conjugate to $\lambda _i$ is denoted by $\sigma _{\lambda_i} ({H_0}) $.
\begin{assumption}
The desired final state ${\left| {{\psi _f}} \right\rangle }$ (and also $\left[ {\left| {{\psi _f}} \right\rangle } \right]$)\footnote{All of the equivalent states of ${\left| {{\psi _f}} \right\rangle }$ are in the same eigenspace of $H_0$ conjugate to ${\lambda _{Hf}}$} is an eigenstate of $H_0$ with eigenvalue ${\lambda _{Hf}}$( i.e.,$\left[ {\left| {{\psi _f}} \right\rangle } \right]\in \sigma _{\lambda_{Hf}} ({H_0})$) and the corresponding eigenspace is degenerate.
\label{A20}
\end{assumption}

\begin{assumption}
The desired final state ${\left| {{\psi _f}} \right\rangle }$ is not an eigenstate of $H_k$ for at least one $k$ i.e., $\exists k;\left[ {\left| {{\psi _f}} \right\rangle } \right] \notin \bigcup\limits_j {{\sigma _{{\lambda _j}}}({H_k})} $.
\label{A30}
\end{assumption}
In the upcoming sections, it will be shown that the control Hamiltonians containing ${\left| {{\psi _f}} \right\rangle }$ as their eigenstate do not contribute to the manipulation process.
\begin{assumption}
The desired final state ${\left| {{\psi _f}} \right\rangle }$  is an eigenstate of $X$ with eigenvalue ${\lambda _{Xf}}$. i.e., $\left[ {\left| {{\psi _f}} \right\rangle } \right]\in \sigma _{\lambda_{Xf}} ({X})$.
\label{A40}
\end{assumption}
\begin{assumption}
The control Hamiltonian set $\left\{ {{H_k}} \right\}$ has at least $n-1$ elements, any of which does not have ${\left| {{\psi _f}} \right\rangle }$ as its eigenket, and the Hamiltonians $\left\{ {{H_0},{H_1},...,{H_{n - 1}}} \right\}$ constitute an $n$-element linearly independent set.
\label{A50}
\end{assumption}
\par In the rest of this article, each of these assumptions will help to derive the desired results. In the next section, the stabilization process and the stochastic boundedness analysis are studied for the formulated problem by the use of extended Lyapunov theory.

\section{Extended stochastic Lyapunov theory}\label{p12}
The process of designing control signals $U$ to achieve the desired goals, is inspired by stochastic Lyapunov theory. This theory is highly dependent on the It\"o formula, which gives the increment of the Lyapunov function according to a continuous Feller (and thus, strong Markov) stochastic process. In the problem of manipulating a quantum system, we are faced with a difficulty in using the conventional It\"o formula. In this section, it is shown that the proposed Lyapunov function is not differentiable when we employ the complex numbers as the field for our Hilbert space.
\subsection{Lyapunov function}
In this article, a Lyapunov function is employed, which is based on maximizing the transition probability to the desired final state:  
\begin{equation}
V(\left| \psi  \right\rangle ) = \frac{1}{2}\left( {1 - {{\left| {\left\langle {{{\psi _f}}}
 \mathrel{\left | {\vphantom {{{\psi _f}} \psi }}
 \right. \kern-\nulldelimiterspace}
 {\psi } \right\rangle } \right|}^2}} \right)
\label{LYAP}
\end{equation}
This Lyapunov function (which is obviously positive) is inspired by the Hilbert-Schmidt norm of an operator. In what follows, the advantages of this Lyapunov candidate will become clear. First, the mathematical properties of (\ref{LYAP}) is discussed within some lemmas:
\begin{lemma}
Consider the Lyapunov function in (\ref{LYAP}), the following statements are equivalent:
\begin{enumerate}[label=(\alph*)]
\item{$V(\left| \psi  \right\rangle ) =0$}\label{a}
\item {$\left| \psi  \right\rangle  = \left[ {\left| {\psi {}_f} \right\rangle } \right]$}\label{b}
\end{enumerate}
\label{L20}
\end{lemma}
\begin{proof}
\ref{b}$\rightarrow$\ref{a} is obvious. For \ref{a}$\rightarrow$\ref{b}, the sufficient condition for \ref{a}  is ${\left| {\left\langle {{{\psi _f}}}
 \mathrel{\left | {\vphantom {{{\psi _f}} \psi }}
 \right. \kern-\nulldelimiterspace}
 {\psi } \right\rangle } \right|^2} = 1$. Using the fact that both $\left| \psi  \right\rangle$ and $\left| \psi_f  \right\rangle$ are normalized, by Cauchy$-$Bunyakovsky$-$Schwarz inequality, the necessary and sufficient condition is $\left| \psi  \right\rangle = c \left| \psi_f  \right\rangle$ with $|c|=1$, which means \ref{b} by D\ref{EQ}.
\end{proof}
This lemma plays an important role in guaranteeing asymptotic stability properties in the rest of this article. By this lemma, vanishing Lyapunov function exclusively describes the convergence of the quantum trajectories to the equivalence class of the desired final state. The following lemma will play an important role in asymptotic stability of the quantum trajectories in upcoming sections.
\begin{lemma}
Consider the Lyapunov function in (\ref{LYAP}), suppose that $\left| {{\psi}} \right\rangle  \notin \left[ {\left| {{\psi _f}} \right\rangle } \right]$ and $0 < R \leqslant \left\| {\left| \psi  \right\rangle  - [\left| {{\psi _f}} \right\rangle] } \right\|<2$ ($\left| \psi  \right\rangle$ is not in an open $R$-neighbourhood of the set $[\left| {{\psi _f}} \right\rangle]$), then $V(\left| \psi  \right\rangle )$ is bounded away from zero i.e., there exists $\nu (R)>0$ such that $\nu (R) \leqslant V(\left| \psi  \right\rangle )$.
\label{L30}
\end{lemma}
\begin{proof}
The assumption reads ${R^2} \leqslant \left\langle {{\psi  - {e^{-i\varepsilon }}{\psi _f}}}
 \mathrel{\left | {\vphantom {{\psi  - {e^{i\varepsilon }}{\psi _f}} {\psi  - {e^{i\varepsilon }}{\psi _f}}}}
 \right. \kern-\nulldelimiterspace}
 {{\psi  - {e^{i\varepsilon }}{\psi _f}}} \right\rangle$ for all real $\epsilon$, which leads to:
\begin{equation*}
{R^2} \leqslant 2 - 2\operatorname{Re} ({e^{ - i\varepsilon }}\left\langle {{{\psi _f}}}
 \mathrel{\left | {\vphantom {{{\psi _f}} \psi }}
 \right. \kern-\nulldelimiterspace}
 {\psi } \right\rangle ).
\end{equation*}
Take $\left\langle {{{\psi _f}}} \mathrel{\left | {\vphantom {{{\psi _f}} \psi }} \right. \kern-\nulldelimiterspace} {\psi } \right\rangle  = r{e^{i\theta }}$, where $r$ and $\theta$ depend on $\left| \psi  \right\rangle$, and choose $\epsilon =-\theta$. Thus, for all admissible $\left| \psi  \right\rangle$, one may write:
\begin{equation*}
r \leqslant 1 - \frac{{{R^2}}}{2}.
\end{equation*}
The above inequality reads:
\begin{equation*}
\mathop {\sup }\limits_{{\text{admissible }}\left| \psi  \right\rangle } \left| {\left\langle {{{\psi _f}}}
 \mathrel{\left | {\vphantom {{{\psi _f}} \psi }}
 \right. \kern-\nulldelimiterspace}
 {\psi } \right\rangle } \right| \leqslant 1 - \frac{{{R^2}}}{2}
\end{equation*}
Also, it can be shown that the supremum takes the RHS value on the boundary of admissible closed set of $\left| \psi  \right\rangle$. Thus, one may define $ \nu (R) \doteq {R^2} - \frac{{{R^4}}}{4}$.
\end{proof}
\subsection{Extended It\"o formula}
The common well-known It\"o formula, gives the increment of a scalar function, of a strong Markov process driven by It\"o form of stochastic differential equations. It is necessary that the function be twice-differentiable. When we treat a complex-field Hilbert space, our definition on the differentiablility changes (The complex function\footnote{This terminology is used in place of "Function of a complex variable" in this article.} must necessarily admit Cauchy-Riemann condition). The sufficient condition for extending the common It\"o formalism to the complex case is the holomorphism of the complex function, which is highly restrictive. The proposed Lyapunov function (\ref{LYAP}), is neither holomorphic, nor even differentiable in its domain (of course employing a holomorphic Lyapunov function is very restrictive and does not necessarily satisfy the mentioned conditions). In this section, the common It\"o formula is extended to undifferentiable complex functions. In order to proceed, the following definitions, which are induced from G\^ateaux differentiation, are presented:
\begin{definition} 
Consider the complex (not necessarily differentiable) function $V(\left| \psi  \right\rangle ):{S^{2n - 1}} \mapsto \mathbb{R}$.
\begin{enumerate}[label=(\roman*)]
\item {If there exists a functional ${\nabla }V:{\mathbb{C}^n} \mapsto \mathbb{C}^n$, independent of $\left|\partial \psi  \right\rangle$ such that the following limit exists for any fixed $ \left|\partial \psi  \right\rangle  \in {\mathbb{C}^n}$:}
\begin{equation}
\mathop {\lim }\limits_{h \downarrow 0} \frac{{V(\left| \psi  \right\rangle  + h\left| {\partial \psi } \right\rangle ) - V(\left| \psi  \right\rangle ) - h\left( {{\nabla }V(\left| \psi  \right\rangle )} \right)\left| {\partial \psi } \right\rangle }}{h} = 0,
\end{equation}
then ${\nabla }V(\left| \psi  \right\rangle )$ is called the \textit{directional gradient} in the direction $\left|\partial \psi  \right\rangle$.

\item{Assuming that ${\nabla }V(\left| \psi  \right\rangle )$ exists, if there exists an operator ${\nabla^{2} }V(\left| \psi  \right\rangle ):{\mathbb{C}^n} \mapsto \mathbb{C}^{n\times n}$}, independent of $\left|\partial \psi  \right\rangle$ such that the following limit exists:
\begin{equation}
\mathop {\lim }\limits_{h \downarrow 0} \frac{{V(\left| \psi  \right\rangle  + h\left| {\partial \psi } \right\rangle ) - V(\left| \psi  \right\rangle ) - h\left( {{\nabla }V(\left| \psi  \right\rangle )} \right)\left| {\partial \psi } \right\rangle  - \frac{{{h^2}}}{2}\left\langle {\partial \psi } \right|\left( {{\nabla ^2}V(\left| \psi  \right\rangle )} \right)\left| {\partial \psi } \right\rangle }}{{{h^2}}} = 0
\label{SO}
\end{equation}
then, ${\nabla^2 }V(\left| \psi  \right\rangle )$ is called the \textit{second order directional gradient} in the direction $\left|\partial \psi  \right\rangle$.

\end{enumerate}
\label{DDG}
\end{definition}
The previous definitions of the directional gradients, pave the way for analyzing the functional properties of (\ref{LYAP}). The directional gradients for (\ref{LYAP}) are evaluated as:
\begin{equation}
\begin{gathered}
  {\nabla }V(\left| \psi  \right\rangle ) =  - \operatorname{Re} (\left\langle {\psi }
 \mathrel{\left | {\vphantom {\psi  {{\psi _f}}}}
 \right. \kern-\nulldelimiterspace}
 {{{\psi _f}}} \right\rangle \left\langle {{\psi _f}} \right|)  \\ 
  {\nabla ^2}V(\left| \psi  \right\rangle ) =  - \left| {{\psi _f}} \right\rangle \left\langle {{\psi _f}} \right| \\ 
\end{gathered}
\label{DG}
\end{equation}
where $\operatorname{Re} (\left\langle . \right|)\doteq \frac{1}{2}\left( {\left\langle {.}
 \mathrel{\left | {\vphantom {.\psi }}
 \right. \kern-\nulldelimiterspace}
 {\psi } \right\rangle  + \left\langle {\psi }
 \mathrel{\left | {\vphantom {\psi .}}
 \right. \kern-\nulldelimiterspace}
 {.}\right\rangle } \right)$ which acts on $\left| \psi  \right\rangle$.  In the next proposition, it is shown that first and second order gradients will help us to exactly evaluate the perturbed Lyapunov function. 
\begin{proposition}
For the proposed Lyapunov function (\ref{LYAP}), define ${d_{\partial \left| \psi  \right\rangle }}V(\left| \psi  \right\rangle )\doteq {V(\left| \psi  \right\rangle  + \left| {\partial \psi } \right\rangle ) - V(\left| \psi  \right\rangle )}$ for sufficiently small  $\left| {\partial \psi } \right\rangle$, then with the directional gardients in (\ref{DG}), the following equality holds:
\begin{equation}
{d_{\left| {\partial \psi } \right\rangle }}V(\left| \psi  \right\rangle ) = \left( {{\nabla }V(\left| \psi  \right\rangle )} \right)\left| {\partial \psi } \right\rangle  + \frac{1}{2}\left\langle {\partial \psi } \right|\left( {{\nabla ^2}V(\left| \psi  \right\rangle )} \right)\left| {\partial \psi } \right\rangle
\label{110}
\end{equation}
\end{proposition}
\begin{proof} 
\[\begin{gathered}
  {d_{\left| {\partial \psi } \right\rangle }}V(\left| \psi  \right\rangle ) = \frac{{ - 1}}{2}\left( {\left\langle {{{\psi _f}}}
 \mathrel{\left | {\vphantom {{{\psi _f}} {\psi  + \partial \psi }}}
 \right. \kern-\nulldelimiterspace}
 {{\psi  + \partial \psi }} \right\rangle \left\langle {{\psi  + \partial \psi }}
 \mathrel{\left | {\vphantom {{\psi  + \partial \psi } {{\psi _f}}}}
 \right. \kern-\nulldelimiterspace}
 {{{\psi _f}}} \right\rangle  - \left\langle {{{\psi _f}}}
 \mathrel{\left | {\vphantom {{{\psi _f}} \psi }}
 \right. \kern-\nulldelimiterspace}
 {\psi } \right\rangle \left\langle {\psi }
 \mathrel{\left | {\vphantom {\psi  {{\psi _f}}}}
 \right. \kern-\nulldelimiterspace}
 {{{\psi _f}}} \right\rangle } \right) \hfill \\
  {\text{                 }} = \frac{{ - 1}}{2}\left( {\left\langle {{{\psi _f}}}
 \mathrel{\left | {\vphantom {{{\psi _f}} \psi }}
 \right. \kern-\nulldelimiterspace}
 {\psi } \right\rangle \left\langle {{\partial \psi }}
 \mathrel{\left | {\vphantom {{\partial \psi } {{\psi _f}}}}
 \right. \kern-\nulldelimiterspace}
 {{{\psi _f}}} \right\rangle  + \left\langle {{{\psi _f}}}
 \mathrel{\left | {\vphantom {{{\psi _f}} {\partial \psi }}}
 \right. \kern-\nulldelimiterspace}
 {{\partial \psi }} \right\rangle \left\langle {\psi }
 \mathrel{\left | {\vphantom {\psi  {{\psi _f}}}}
 \right. \kern-\nulldelimiterspace}
 {{{\psi _f}}} \right\rangle  + \left\langle {{{\psi _f}}}
 \mathrel{\left | {\vphantom {{{\psi _f}} {\partial \psi }}}
 \right. \kern-\nulldelimiterspace}
 {{\partial \psi }} \right\rangle \left\langle {{\partial \psi }}
 \mathrel{\left | {\vphantom {{\partial \psi } {{\psi _f}}}}
 \right. \kern-\nulldelimiterspace}
 {{{\psi _f}}} \right\rangle } \right) \hfill \\
  {\text{                 }} = \left( {\nabla V(\left| \psi  \right\rangle )} \right)\left| {\partial \psi } \right\rangle  + \frac{1}{2}\langle \partial \psi |\left( {{\nabla ^2}V(\left| \psi  \right\rangle )} \right)\left| {\partial \psi } \right\rangle. \hfill \\ 
\end{gathered} \]
\end{proof}
This proposition presents the directional counterpart of the Taylor series for an undifferentiable real-valued complex function (This is why the directional gradients were defined. Although this Lyapunov function is not analytic, a Taylor-like serie is presented). Also this proposition shows why this Lyapunov candidate suits the problem of manipulating the SSE. The value of perturbed Lyapunov function is exactly evaluated by two perturbation steps.
\par Based on the presented definitions, in the next theorem, the extended version of the It\"o formula for undifferentiable complex functions of a strong Markov process is presented. This extended version of the It\"o formula is the starting point to use the stochastic Lyapunov theory.
\begin{theorem}
Let $V:{S^{2n - 1}} \mapsto \mathbb{R}$ be a real function with well-defined directional gradients defined in D\ref{DDG}. Also, assume that (\ref{110}) holds and $\left| \psi  \right\rangle$ be an adapted stochastic process driven by the following It\"o drift-diffusion stochastic differential equation:
\begin{equation}
d\left| \psi  \right\rangle  = f(\left| \psi  \right\rangle )dt + g(\left| \psi  \right\rangle )dW.
\label{ITO}
\end{equation}
Then, the following equality holds for the stochastic increment of $V$:
\begin{eqnarray}
{d_{d\left| \psi  \right\rangle }}V(\left| \psi  \right\rangle ) &&= \left( {{\nabla }V(\left| \psi  \right\rangle )f(\left| \psi  \right\rangle ) + \frac{1}{2}{{\left( {g(\left| \psi  \right\rangle )} \right)}^\dag }{\nabla ^2}V(\left| \psi  \right\rangle )\left( {g(\left| \psi  \right\rangle )} \right)} \right)dt \nonumber \\&&+ \left( {{\nabla }V(\left| \psi  \right\rangle )g(\left| \psi  \right\rangle )} \right)dW.
\label{EITO}
\end{eqnarray}

\end{theorem}
\begin{proof}
Substituting (\ref{ITO}) into (\ref{110}) in place of ${\left| {\partial \psi } \right\rangle }$ gives:
\begin{equation}
\begin{gathered}
  {d_{d\left| \psi  \right\rangle }}V(\left| \psi  \right\rangle ) = {\nabla }V(\left| \psi  \right\rangle )f(\left| \psi  \right\rangle )dt + {\nabla }V(\left| \psi  \right\rangle )g(\left| \psi  \right\rangle )dW \\ 
  {\text{     }} + \frac{1}{2}\left( \begin{gathered}
  {\left( {f(\left| \psi  \right\rangle )} \right)^\dag }{\nabla ^2}V(\left| \psi  \right\rangle )\left( {f(\left| \psi  \right\rangle )} \right){(dt)^2} \hfill \\
   + {\left( {g(\left| \psi  \right\rangle )} \right)^\dag }{\nabla ^2}V(\left| \psi  \right\rangle )\left( {g(\left| \psi  \right\rangle )} \right){(dW)^2} \hfill \\
   + {\left( {g(\left| \psi  \right\rangle )} \right)^\dag }{\nabla ^2}V(\left| \psi  \right\rangle )\left( {f(\left| \psi  \right\rangle )} \right)(dtdW) \hfill \\
   + {\left( {f(\left| \psi  \right\rangle )} \right)^\dag }{\nabla ^2}V(\left| \psi  \right\rangle )\left( {g(\left| \psi  \right\rangle )} \right)(dtdW) \hfill \\ 
\end{gathered}  \right). \\ 
\end{gathered}
\end{equation}
Keeping terms up to $(dW)^2$ and integrating in the sense of It\"o, gives:
\begin{equation}
\begin{gathered}
  V(\left| {\psi (t)} \right\rangle ) = \int_0^t {{\nabla }V(\left| \psi  \right\rangle )f(\left| \psi  \right\rangle )dt}  + \int_0^t {{\nabla }V(\left| \psi  \right\rangle )g(\left| \psi  \right\rangle )dW}  \\ 
   + \frac{1}{2}\int_0^t {{{\left( {g(\left| \psi  \right\rangle )} \right)}^\dag }{\nabla ^2}V(\left| \psi  \right\rangle )\left( {g(\left| \psi  \right\rangle )} \right){{(dW)}^2}}  \\ 
\end{gathered}
\end{equation}
Using the It\"o multiplication rule for the third integral gives \cite{chen1995linear}:
\begin{equation}
\begin{gathered}
  V(\left| {\psi (t)} \right\rangle ) = \int_0^t {\left( {{\nabla }V(\left| \psi  \right\rangle )f(\left| \psi  \right\rangle ) + \frac{1}{2}{{\left( {g(\left| \psi  \right\rangle )} \right)}^\dag }{\nabla ^2}V(\left| \psi  \right\rangle )\left( {g(\left| \psi  \right\rangle )} \right)} \right)dt}  \\ 
   + \int_0^t {{\nabla }V(\left| \psi  \right\rangle )g(\left| \psi  \right\rangle )dW}. \\ 
\end{gathered} 
\end{equation}
Therefore, the It\"o form of stochastic differential equation for $V(\left| \psi  \right\rangle )$ is:
\begin{equation}
\begin{gathered}
  {d_{d\left| \psi  \right\rangle }}V(\left| \psi  \right\rangle ) = \left( {{\nabla}V(\left| \psi  \right\rangle )f(\left| \psi  \right\rangle ) + \frac{1}{2}{{\left( {g(\left| \psi  \right\rangle )} \right)}^\dag }{\nabla ^2}V(\left| \psi  \right\rangle )\left( {g(\left| \psi  \right\rangle )} \right)} \right)dt \\ 
   + \left( {{\nabla }V(\left| \psi  \right\rangle )g(\left| \psi  \right\rangle )} \right)dW. \\ 
\end{gathered}
\end{equation}
\end{proof}
In the rest of this paper, the coefficient of $dt$ will be denoted by ${L_{d\left| \psi  \right\rangle }}V(\left| \psi  \right\rangle )$ or ${L}V(\left| \psi  \right\rangle )$.
\begin{remark}
The previous theorem presents the stochastic increment of $V(\left| \psi  \right\rangle )$. This increment is evaluated by keeping terms up to $O(dt)$, which is reasonable since $dt\rightarrow 0$. Also, for the complex functions with directional gradients not satisfying (\ref{110}), one can derive (\ref{EITO}) in the second order approximation, which is valid as $dt\rightarrow 0$.
\end{remark}
\par For the Lyapunov function (\ref{LYAP}), with directional gradients in (\ref{DG}), one may deduce that the stochastic increment of $V(\left| \psi  \right\rangle )$ regarding the SSE defined by (\ref{SSE}) and driven by the control field in (\ref{HU}) has the following form:
\begin{equation}
\begin{gathered}
  {d_{d\left| \psi  \right\rangle }}V(\left| \psi  \right\rangle ) = \left( \begin{gathered}
  \frac{{ - 1}}{\hbar }\operatorname{Im}\left( {\left\langle {\psi }
 \mathrel{\left | {\vphantom {\psi  {{\psi _f}}}}
 \right. \kern-\nulldelimiterspace}
 {{{\psi _f}}} \right\rangle \left\langle {{\psi _f}} \right|H(U)\left| \psi  \right\rangle } \right) \\
 + k\operatorname{Re} \left( {\left\langle {\psi }
 \mathrel{\left | {\vphantom {\psi  {{\psi _f}}}}
 \right. \kern-\nulldelimiterspace}
 {{{\psi _f}}} \right\rangle \left\langle {{\psi _f}} \right|{{\left( {X - \left\langle X \right\rangle } \right)}^2}\left| \psi  \right\rangle } \right) \hfill \\
   - k{\left| {\left\langle {{\psi _f}} \right|\left( {X - \left\langle X \right\rangle } \right)\left| \psi  \right\rangle } \right|^2} \hfill \\ 
\end{gathered}  \right)dt \\ 
   - \sqrt {2k} \operatorname{Re} \left( {\left\langle {\psi }
 \mathrel{\left | {\vphantom {\psi  {{\psi _f}}}}
 \right. \kern-\nulldelimiterspace}
 {{{\psi _f}}} \right\rangle \left\langle {{\psi _f}} \right|\left( {X - \left\langle X \right\rangle } \right)\left| \psi  \right\rangle } \right)dW. \\ 
\end{gathered}
\label{IG}
\end{equation}
Based on the proposed mathematical background, in the following sections, main results on the stability of quantum trajectories are presented.

\section{Main results on the stability of the SSE}\label{p13}
\subsection{Definitions}
The stochastic properties, studied in this article are constructed on the complete probability space $(\Omega,\Xi ,\operatorname{P})$, where $\Omega$ is the sample space, $\Xi$ is the $\sigma$-algebra generated by $\Omega$ and $\operatorname{P}$ is a probability measure defined on $\Xi$. Let us consider a quantum stochastic process, which is the solution of SSE (\ref{SSE}), by $\left| {{\psi }(\omega ,t)} \right\rangle :\left( {\Omega  \times T,\Xi  \otimes {{\rm B}(T) }} \right) \mapsto \left( {{\mathcal{H}},{\rm B}({\mathcal{H}})} \right)$ and a quantum trajectory (a sample path of $\left| {{\psi }(\omega ,t)} \right\rangle$) started from $\left| \psi_0 \right\rangle$ at $t=0$ by $\left| {{\psi ^{\left| {{\psi _0}} \right\rangle }}(t)} \right\rangle :T \mapsto \mathcal{H}$, where $T \doteq \left[ {0,\infty } \right)$ is the time index and $\rm B(.)$ denotes the Borel $\sigma$-algebra generated by the corresponding set.
\par Lemma \ref{L20} provides that support of $V$ ($\operatorname{supp}(V) =\left\{ {\left| \psi  \right\rangle  \notin \left[ {\left| {{\psi _f}} \right\rangle } \right]} \right\}$), is compact in ${\mathbb{C}^n}$, thus ${L}V(\left| \psi  \right\rangle )$ is the infinitesimal generator of $V(\left| \psi  \right\rangle )$, i.e.,
\begin{equation}
LV(\left| {{\psi ^{\left| {{\psi _0}} \right\rangle }}(t)} \right\rangle ) = \mathop {\lim }\limits_{h \downarrow 0} \frac{{\operatorname{E} \left[ {V\left( {\left| {\psi (\omega ,t + h)} \right\rangle } \right)} \right] - V\left( {\left| {{\psi ^{\left| {{\psi _0}} \right\rangle }}(t)} \right\rangle } \right)}}{h}.
\end{equation}
\par Based on the previous sections, we propose two definitions on the stability of quantum trajectories.
\begin{definition}
The quantum stochastic process $\left| {\psi (\omega ,t)} \right\rangle  \equiv \left[ {\left| {{\psi _f}} \right\rangle } \right]$ driven by (\ref{SSE}) is said to be:
\begin{enumerate}[label=(\roman*)]
\item { \textit{stochastically stable} if for any $0<\epsilon$:}
\begin{equation}
\mathop {\lim }\limits_{\left\| {\left| {\partial \psi } \right\rangle } \right\| \to 0}\operatorname{P}\left\{ {\mathop {\sup }\limits_{0 \leqslant t} \left\| {\left| {{\psi ^{\left| {{\psi _f}} \right\rangle  + \left| {\partial \psi } \right\rangle }}(t)} \right\rangle  - \left[ {\left| {{\psi _f}} \right\rangle } \right]} \right\| \geqslant \varepsilon } \right\} = 0
\end{equation}
\item {\textit{ stochastically asymptotically stable} if it is stochastically stable and also:}
\begin{equation}
\mathop {\lim }\limits_{\left\| {\left| {\partial \psi } \right\rangle } \right\| \to 0} P\left\{ {\mathop {\sup }\limits_{0 \leqslant t} \left\| {\left| {{\psi ^{\left| {{\psi _f}} \right\rangle  + \left| {\partial \psi } \right\rangle }}(t)} \right\rangle  - \left[ {\left| {{\psi _f}} \right\rangle } \right]} \right\| = 0} \right\} = 1.
\end{equation}
\end{enumerate}
\end{definition}
The previous definitions presented two notions of stability for the considered quantum trajectories. A stable quantum trajectory remains in a neighbourhood of the equivalence class of the desired final state almost surely as the initial value approaches that equivalence class. The second notion guarantees that the quantum trajectory does not escape the equivalence class almost surely.

\subsection{Stochastic stability of SSE}
\par In order to study the stability of quantum trajectories, the conditions $\hat f(\left[ {\left| {{\psi _f}} \right\rangle } \right]) = \hat g(\left[ {\left| {{\psi _f}} \right\rangle } \right]) = {\mathbf{0}}$ must necessarily hold. These conditions are satisfied if A\ref{A20} and A\ref{A40} are satisfied in addition to the condition that control signals $U$ vanish at $\left[ {\left| {{\psi _f}} \right\rangle } \right]$. Thus A\ref{A20} and A\ref{A40} are the basic necessary assumptions in the rest of this section. Also one may deduce that according to A\ref{A20} and A\ref{A40}, the infinitesimal generator ${L}V(\left| \psi  \right\rangle )$ in (\ref{IG}) takes the following form:
\begin{equation}
LV(\left| \psi  \right\rangle ) =  - \frac{1}{\hbar }\sum\limits_{k = 1}^n {{u_k}(t)} \operatorname{Im} \left( {\left\langle {\psi }
 \mathrel{\left | {\vphantom {\psi  {{\psi _f}}}}
 \right. \kern-\nulldelimiterspace}
 {{{\psi _f}}} \right\rangle \left\langle {{\psi _f}} \right|{H_k}\left| \psi  \right\rangle } \right).
\label{LV1}
\end{equation}
Now the necessity of A\ref{A30} is more obvious. In the absence of A\ref{A30}, the SSE (\ref{SSE}) would be uncontrollable since all the coefficients of control signals in (\ref{LV1}) would vanish everywhere in ${S^{2n - 1}}$. In the rest of this paper, the following control signals will be used:
\begin{equation}
{u_k}(t) = {\alpha _k}\operatorname{Im} \left( {{e^{i\measuredangle \left\langle {\psi }
 \mathrel{\left | {\vphantom {\psi  {{\psi _f}}}}
 \right. \kern-\nulldelimiterspace}
 {{{\psi _f}}} \right\rangle }}\left\langle {{\psi _f}} \right|\left. {{H_k}} \right|\left. \psi  \right\rangle } \right),
\label{control}
\end{equation}
where ${\alpha _k} \in {\mathbb{R}^ + }$. This control signal has also been used for stabilizing the deterministic Schr\"odinger equation \cite{shuang2007quantum}. Using the proposed control signals yields:
\begin{equation}
LV(\left| \psi  \right\rangle ) =  - \frac{1}{\hbar }\sum\limits_{k = 1}^n {{\alpha _k}} \left| {\left\langle {\psi }
 \mathrel{\left | {\vphantom {\psi  {{\psi _f}}}}
 \right. \kern-\nulldelimiterspace}
 {{{\psi _f}}} \right\rangle } \right|{\left( {\operatorname{Im} \left( {{e^{i\measuredangle \left\langle {\psi }
 \mathrel{\left | {\vphantom {\psi  {{\psi _f}}}}
 \right. \kern-\nulldelimiterspace}
 {{{\psi _f}}} \right\rangle }}\left\langle {{\psi _f}} \right|{H_k}\left| \psi  \right\rangle } \right)} \right)^2} \leqslant 0.
\label{LV2}
\end{equation}
Now, the following lemma can be stated:
\begin{lemma}
\label{L4}
Consider the Lyapunov function (\ref{LYAP}) where the quantum dynamic is driven by (\ref{SSE}) with control signals (\ref{control}). Then the process $V(\left| \psi  \right\rangle )$ is a \textit{supermartingale}. Also, $\mathop {\lim }\limits_{t \to \infty } \operatorname{E} [V(\left| \psi  \right\rangle )]$ exists and is equal to $\operatorname{E} [V(\mathop {\lim }\limits_{t \to \infty } \left| {{\psi ^{\left| {{\psi _f}} \right\rangle }}(\omega ,t)} \right\rangle )]$.
\end{lemma}
\begin{proof}
The fact that SSE (\ref{SSE}) is a unitary evolution induces that $\operatorname{P} \left\{ {{\tau _m} < \infty } \right\} = 1$, where $\tau _m$ is the first exit time from ${Q_m} \doteq \left\{ {\left| \psi  \right\rangle \left| {V(\left| \psi  \right\rangle ) < m} \right.} \right\}$, i.e., ${\tau _m} \doteq \inf \left\{ {t\left| {\left| {{\psi ^{\left| {{\psi _0}} \right\rangle }}(t)} \right\rangle  \notin {Q_m}} \right.} \right\}$ with $\left| {{\psi _0}} \right\rangle  \in {Q_m}$ almost surely and $1<m$. Thus, $\tau \wedge t = t$ ($\tau \wedge t$ denotes $ \min (\tau ,t) $) almost surely for all $t$. Also, take ${\Gamma _t}$ the family of $\sigma$-algebras of sets of $\Xi$ generated by the Wiener processes $dW$up to time $t$. Now using Dynkins formula gives \cite{kushner1967stochastic}:
\begin{equation}
\operatorname{E} [V\left( {\left| {\psi \left( {\omega ,\left( {t + h} \right) \wedge {\tau _m}} \right)} \right\rangle } \right)\left| {{\Gamma _t}} \right.] = V\left( {\left| {\psi (\omega ,t)} \right\rangle } \right) + \int_t^{(t + h)\wedge {\tau _m}} {LV\left( {\left| {\psi (\omega ,u)} \right\rangle } \right)} du.
\end{equation}
Thus,
\begin{equation}
\operatorname{E} [V\left( {\left| {\psi \left( {\omega ,t + h} \right)} \right\rangle } \right)\left| {{\Gamma _t}} \right.] \leqslant V\left( {\left| {\psi (\omega ,t)} \right\rangle } \right)
\end{equation}
which results the supermartingale property. The fact that $V(\left| \psi  \right\rangle )$ is a non-negative supermartingale gives that $\mathop {\lim }\limits_{t \to \infty } \operatorname{E} [V(\left| \psi  \right\rangle )]$ exists and is equal to \\$\operatorname{E} [V(\mathop {\lim }\limits_{t \to \infty } \left| {{\psi ^{\left| {{\psi _f}} \right\rangle }}(\omega ,t)} \right\rangle )]$ \cite{doob1953stochastic}.
\end{proof}
Now take $0 < R < 2$ and define ${N_R} \doteq \left\{ {\left| \psi  \right\rangle \left| {\left\| {\left| \psi  \right\rangle  - \left| {{\psi _f}} \right\rangle } \right\|} \right. \leqslant R} \right\}$. Assume that $\left| {{\psi _0}} \right\rangle  = \left| {{\psi _f}} \right\rangle  + \left| {\partial \psi } \right\rangle  \in {N_R}$ almost surely. If ${\tau _{{N_R}}} \doteq \inf \left\{ {t\left| {\left| {{\psi ^{\left| {{\psi _0}} \right\rangle }}\left( t \right)} \right\rangle  \notin } \right.{N_R}} \right\}$ be the first exit time from $N_R$, by Lemma \ref{L4} one may deduce:
\begin{equation}
\operatorname{E} [V\left( {\left| {\psi \left( {\omega ,{\tau _{{N_R}}} \wedge t} \right)} \right\rangle } \right)\left| {{\Gamma _t}} \right.] \leqslant V\left( {\left| {{\psi _0}} \right\rangle } \right)
\label{270}
\end{equation}
which expresses that the stopped process $V\left( {\left| {\psi \left( {\omega ,{\tau _{{N_R}}} \wedge t} \right)} \right\rangle } \right)$ is also a supermartingale.
Based on these derivations, the following theorem is presented, which plays an important role in stabilization procedure and shows the stochastic stability of ${\left| {{\psi _f}} \right\rangle }$ .
\begin{theorem}
\label{TH20}
Consider the SSE (\ref{SSE}) with assumptions discussed in section \ref{COQS}. Assume that A\ref{A10} to A\ref{A40} hold. With control signals (\ref{control}), the quantum stochastic process $\left| {\psi (\omega ,t)} \right\rangle  \equiv \left[ {\left| {{\psi _f}} \right\rangle } \right]$ is stochastically stable.
\end{theorem}
\begin{proof}
Assume that $\left| {{\psi _0}} \right\rangle  = \left| {{\psi _f}} \right\rangle  + \left| {\partial \psi } \right\rangle  \in {N_R}$ for some $0 < R < 2$ . By Lemma \ref{L4} and (\ref{270}), one writes:
\begin{equation}
\label{280}
\operatorname{E} [\mathop {\sup}\limits_{0 < t} V\left( {\left| {{\psi ^{\left| {{\psi _0}} \right\rangle }}\left( {{\tau _{{N_R}}} \wedge t} \right)} \right\rangle } \right)] \leqslant V\left( {\left| {{\psi _0}} \right\rangle } \right).
\end{equation}
Defining $y(w) \doteq \mathop {\sup }\limits_{0 < t} \left\| {\left| {{\psi ^{\left| {{\psi _0}} \right\rangle  }}\left( {{\tau _{{N_R}}} \wedge t} \right)} \right\rangle  - \left[ {\left| {{\psi _f}} \right\rangle } \right]} \right\|$ and $\Omega ' \doteq \left\{ {\omega \left| {R \leqslant y(\omega )} \right.} \right\}$, (\ref{280}) can be rewritten as:
\begin{equation}
\begin{gathered}
  V\left( {\left| {{\psi _0}} \right\rangle } \right) \geqslant \int\limits_\Omega  {\mathop {\sup }\limits_{0 < t} V\left( {\left| {{\psi ^{\left| {{\psi _0}} \right\rangle }}\left( {{\tau _{{N_R}}} \wedge t} \right)} \right\rangle } \right)d\operatorname{P} } (\omega ) \hfill \\
   \hspace{1.3cm}\geqslant \int\limits_{\Omega '} {\mathop {\sup }\limits_{0 < t} V\left( {\left| {{\psi ^{\left| {{\psi _f}} \right\rangle  + \left| {\partial \psi } \right\rangle }}\left( {{\tau _{{N_R}}} \wedge t} \right)} \right\rangle } \right)d\operatorname{P} } (\omega ) \hfill \\
   \hspace{1.3cm}\geqslant \left( \mathop {\inf }\limits_{\Omega '}  {\mathop {\sup }\limits_{0 < t} V\left( {\left| {{\psi ^{\left| {{\psi _f}} \right\rangle  + \left| {\partial \psi } \right\rangle }}\left( {{\tau _{{N_R}}} \wedge t} \right)} \right\rangle } \right)} \right)\operatorname{P} \left\{ {y > R} \right\}. \hfill \\ 
\end{gathered}
\end{equation}
By Lemma \ref{L30} one may deduce:
\begin{equation}
\operatorname{P} \left\{ {\mathop {\sup }\limits_t \left\| {\left| {{\psi ^{\left| {{\psi _f}} \right\rangle  + \left| {\partial \psi } \right\rangle }}\left( t \right)} \right\rangle  - \left[ {\left| {{\psi _f}} \right\rangle } \right]} \right\| > R} \right\} \leqslant \frac{{V\left( {\left| {{\psi _0}} \right\rangle } \right)}}{{\nu (R)}}.
\end{equation}
Now, Lemma \ref{L20} and the continuity of $V(\left| \psi  \right\rangle )$ give:
\begin{equation}
\mathop {\lim }\limits_{\left\| {\left| {\partial \psi } \right\rangle } \right\| \to 0} \operatorname{P} \left\{ {\mathop {\sup }\limits_t \left\| {\left| {{\psi ^{\left| {{\psi _f}} \right\rangle  + \left| {\partial \psi } \right\rangle }}\left( t \right)} \right\rangle  - \left[ {\left| {{\psi _f}} \right\rangle } \right]} \right\| > R} \right\} \leqslant \mathop {\lim }\limits_{\left\| {\left| {\partial \psi } \right\rangle } \right\| \to 0} \frac{{V\left( {\left| {{\psi _0}} \right\rangle } \right)}}{{\nu (R)}} = 0.
\label{310}
\end{equation}
\end{proof}
Theorem \ref{TH20} reveals that the desired final state ${\left| {{\psi _f}} \right\rangle }$ is stochastically stable and the trajectories remain in any prescribed neighbourhood of ${\left[ {\left| {{\psi _f}} \right\rangle } \right]}$ with probability $1$. Also, $\left| {{\psi _0}} \right\rangle  \in \left[ {\left| {{\psi _f}} \right\rangle } \right]$ is a special case of this result. In the rest of this section, we will show that under some conditions on the control Hamiltonians in (\ref{HU}), the stochastic asymptotic stability can be achieved.

\subsection{Stochastic Asymptotic stability of SSE}
Theorem \ref{TH20} represented a stability condition on the quantum trajectories based on (\ref{LV2}). In this sense, the system would evolve until reaching its invariant set, which is itself a subset of $\left\{ {\left[ {\left| \psi  \right\rangle } \right]\left| {LV(\left| \psi  \right\rangle ) = 0} \right.} \right\}$. In order to characterize this set, let us denote the set of eigenvalues of control Hamiltonians by $\sigma ({H}) \doteq \bigcup\limits_k {\sigma ({H_k})}$. Also, based on the Cartan decomposition of $su(n)$, one can always find a basis, in which $H_0$ is diagonal. Denote this basis by $\left\{ {\left| 1 \right\rangle ,\left| 2 \right\rangle ,...,\left| n \right\rangle } \right\}$, where the bases are mutually orthogonal. Without loss of generality, assume that the eigenspace corresponding to ${\left| 1 \right\rangle }$ is degenerate and $\left| {{\psi _f}} \right\rangle  = \left| 1 \right\rangle$. Now using (\ref{LV2}), the following theorem can be stated:
\begin{theorem}
\label{TH30}
Consider the state dynamics (\ref{SSE}) and the Lyapunov function (\ref{LYAP}), also assume that A\ref{A10} to A\ref{A50} hold. The set of quantum states, in which $LV(\left| \psi  \right\rangle ) = 0$, can be decomposed into the following two subsets (i.e. $\left\{ {\left| \psi  \right\rangle \left| {LV(\left| \psi  \right\rangle ) = 0} \right.} \right\} = A \cup B $):
\begin{itemize}
\item {$A = {\left[ {\left| {{\psi _f}} \right\rangle } \right]^ \bot }$} 
\item {$B = \left\{ {\left| \psi  \right\rangle \left| {\forall k, \exists {\lambda _k} \in \mathbb{R}: \left\langle {{{\psi _f}}}
 \mathrel{\left | {\vphantom {{{\psi _f}} {{H_k} - {\lambda _k}I\left| \psi  \right.}}}
 \right. \kern-\nulldelimiterspace}
 {{{H_k} - {\lambda _k}I\left| \psi  \right.}} \right\rangle  = 0} \right.} \right\}$}
\\ Also, the following statements hold:
\begin{enumerate}[label=(B.\roman*)]
\item {In the case that the control Hamiltonians $\left\{H_k\right\}$ have no common eigenkets, $B$ includes at most one equivalence class of states for each choise of $\left\{ {{\lambda _k} \in \mathbb{R},k = 1,...,n - 1} \right\}$\label{B1}}.
\item{\label{B2}In the case that the control Hamiltonians $\left\{H_k\right\}$ have $s$ independent common eigenkets (each of them is an eigenket for at least $2$ of the Hamiltonians) then $B$ includes at most $1+s$ different equivalence classes of quantum states.}
\item{\label{B3}$\left[ {\left| {{\psi _f}} \right\rangle } \right] \in B$ for some choice of $\left\{ {{\lambda _k} \in \mathbb{R} - \sigma ({H})} \right\}$.}
\end{enumerate}
\end{itemize}
\end{theorem}
Before proceeding with the proof, let us present a lemma which will be used in the proof of this theorem.

\begin{lemma}
\label{L5}
Consider the control Hamiltonians in (\ref{HU}). Assume that A\ref{A10} to A\ref{A50} hold. Then, for each choice of $\left\{ {{\lambda _k} \in \mathbb{R},k = 1,...,n - 1} \right\}$, the set \\$\left\{{ H_1}-\lambda_1I, { H_2}-\lambda_2I,...,{ H_{n-1}}-\lambda_{n-1}I \right\}$ is linearly independent.
\end{lemma}
\begin{proof}
Without loss of generality assume that $\lambda_1 \ne 0$. By contradiction assume that $$\left\{{ H_1}-\lambda_1I, { H_2}-\lambda_2I,...,{ H_{n-1}}-\lambda_{n-1}I \right\}$$ is linearly dependent. Thus for some nonzero set $\left\{ c_2,...,{ c_{n-1}}\right\}$ one may find an scalar $\alpha$ such that:
\begin{equation*}
\alpha ({H_1} - {\lambda _1}I) = \sum\limits_{k = 2}^{n - 1} {{c_k}({H_k} - {\lambda _k}I)}  \Rightarrow \alpha {H_1} = \sum\limits_{k = 2}^{n - 1} {\left( {{c_k}{H_k}} \right)}  - \sum\limits_{k = 2}^{n - 1} {\left( {{c_k}{\lambda _k}} \right)} I + \alpha {\lambda _1}I.
\end{equation*}
By the fact that $i{H_k} \in su(n)$ (and thus they are traceless), one has the unique choice of $\alpha  = \frac{{\sum\limits_{k = 2}^{n - 1} {\left( {{c_k}{\lambda _k}} \right)} }}{{{\lambda _1}}}$.  So one deduces that:
\begin{equation}
\alpha {H_1} = \sum\limits_{k = 2}^{n - 1} {\left( {{c_k}{H_k}} \right)}
\end{equation}
which contradicts the assumption A\ref{A50} in both cases $\alpha=0$ and $\alpha \ne 0$.
\end{proof}
\begin{proof}[proof of Theorem \ref{TH30}]
Due to (\ref{LV2}), the space perpendicular to ${\left[ {\left| {{\psi _f}} \right\rangle } \right]^ \bot }$ belongs to $\left\{ {\left| \psi  \right\rangle \left| {LV(\left| \psi  \right\rangle ) = 0} \right.} \right\} $, which shows $A$. In this proof, first we neglect the unitarity of the quantum state, and after finding the un-normalized solution subspace, it will be intersected with the unit sphere. Assume that $\mathbb{R}$ is partitioned as $\mathbb{R} = \left( {\mathbb{R} - \sigma ({H_k})} \right) \cup \sigma ({H_k})$. If ${\left| {{\psi _0}} \right\rangle } \notin {\left[ {\left| {{\psi _f}} \right\rangle } \right]^ \bot }$, (\ref{LV2}) implies that we should search for the common solutions of ${\operatorname{Im} \left( {{e^{i\measuredangle \left\langle {\psi }
 \mathrel{\left | {\vphantom {\psi  {{\psi _f}}}}
 \right. \kern-\nulldelimiterspace}
 {{{\psi _f}}} \right\rangle }}\left\langle {{\psi _f}} \right|{H_k}\left| \psi  \right\rangle } \right)}=0$ for all $k$, but:
\begin{equation}
\operatorname{Im} \left( {{e^{i\measuredangle \left\langle {{{\psi _f}}}
 \mathrel{\left | {\vphantom {{{\psi _f}} \psi }}
 \right. \kern-\nulldelimiterspace}
 {\psi } \right\rangle }}\langle {\psi _f}|{H_k}\left| \psi  \right\rangle } \right) = 0 \Leftrightarrow \langle {\psi _f}|{H_k}\left| \psi  \right\rangle  = {\lambda _k}\left\langle {{{\psi _f}}}
 \mathrel{\left | {\vphantom {{{\psi _f}} \psi }}
 \right. \kern-\nulldelimiterspace}
 {\psi } \right\rangle \Leftrightarrow \langle {\psi _f}|{H_k} - {\lambda _k}I\left| \psi  \right\rangle  = 0
\label{320}
\end{equation}
for real $\lambda _k$'s.
\\ Let us first prove \ref{B1}. Assume that $ {{\lambda _k} \in \mathbb{R} - \sigma ({H_k})}$ for all $k$. Thus, ${H_k} - {\lambda _k}I$ is non-singular. Hence, we may characterize the subspace, which the solutions of (\ref{320}) belong to for each $k$ as:
\begin{equation}
{S_k}(\lambda_k) \doteq \operatorname{span}\left\{ {{{\left( {{H_k} - {\lambda _k}I} \right)}^{ - 1}}\left| 2 \right\rangle ,...,{{\left( {{H_k} - {\lambda _k}I} \right)}^{ - 1}}\left| n \right\rangle } \right\},
\end{equation}
which is an $(n-1)-$dimensional subspace regarding to linear independence of $\left| i \right\rangle$'s. Thus, the solution of (\ref{320}) must necessarily belong to the intersection of ${S_k}(\lambda_k)$'s for each choice of $\left\{ {{\lambda _k} \in \mathbb{R} - \sigma ({H_k})} \right\}$:
\begin{equation}
\left| \psi  \right\rangle  \in  {\bigcap\limits_k {{S_k}({\lambda _k})} } .
\end{equation}
By A\ref{A50} and Lemma \ref{L5}, we may deduce that none of the subspaces ${S_k}(\lambda_k)$ can exactly coincide. For further demonstrations, one may show that \\$\sum\limits_{j = 2}^n {\left\langle j \right|({H_t} - {\lambda _t}I){{({H_u} - {\lambda _u}I)}^{ - 1}}\left| 1 \right\rangle } \left| j \right\rangle$ belongs to ${S_t}(\lambda_t)$ but not ${S_u}(\lambda_u)$ for each pair of distinct $u$ and $t$ $\in \left \{1,2,...,n-1 \right\}$. Also, for each distinct $s$, $u$, and $t$, $\sum\limits_{j = 2}^n {\left\langle j \right|({H_t} - {\lambda _t}I){{({H_u} - {\lambda _u}I)}^{ - 1}}\left| 1 \right\rangle } \left| j \right\rangle$ and $\sum\limits_{j = 2}^n {\left\langle j \right|({H_s} - {\lambda _s}I){{({H_u} - {\lambda _u}I)}^{ - 1}}\left| 1 \right\rangle } \left| j \right\rangle$ cannot be collinear (based on Lemma \ref{L5}). The dimension of intersection of $n-1$ non-coincident $(n-1)-$dimensional subspaces is not more than $1$. Now, intersecting the $1$-dimensional solution subspace with the unit sphere implies that $B$ includes at most one equivalence class of quantum states for each choice of $\left\{ {{\lambda _k} \in \mathbb{R} - \sigma ({H_k})} \right\}$. \\Based on this proof, choosing ${\lambda _k} = \langle {\psi _f}|{H_k}\left| {{\psi _f}} \right\rangle$, results in \ref{B3}.
\\ Now, assume that $ {{\lambda _k} \in \sigma ({H_k})}$ for some $k$'s but not all of them. In this case, the solution spaces for (\ref{320}), ($S_k(\lambda _k)$) are defined in a more general manner in order to include singular $({H_k} - {\lambda _k}I)$'s. First, define the non-homogeneous part of $S_k(\lambda _k)$ as:
\begin{eqnarray*}
nh{S_k}({\lambda _k}) \doteq &&\left\{ {\left| \psi  \right\rangle \left| {({H_k} - {\lambda _k}I)} \right.\left| \psi  \right\rangle  = \left| 2 \right\rangle } \right\} \\&&\cup \left\{ {\left| \psi  \right\rangle \left| {({H_k} - {\lambda _k}I)} \right.\left| \psi  \right\rangle  = \left| 3 \right\rangle } \right\} \\&&\cup ... \cup \left\{ {\left| \psi  \right\rangle \left| {({H_k} - {\lambda _k}I)} \right.\left| \psi  \right\rangle  = \left| n \right\rangle } \right\}
\end{eqnarray*}
which includes at most $(n-1)-$deg$(\lambda _k,H_k)$ independent vectors, where (deg$(\lambda _k,H_k)$ is the degeneracy of $\lambda _k$ for $H_k$. Also, define the homogeneous part $h{S_k}({\lambda _k})$ to be the kernel of $({H_k} - {\lambda _k}I)$. Now, the solution space can be defined as:
\begin{equation}
{S_k}({\lambda _k}) \doteq {\text{span}}\left\{ {h{S_k}({\lambda _k}),nh{S_k}({\lambda _k})} \right\},
\end{equation}
which is at most $(n-1)-$dimensional. Consider $s$, $t$, and $u$ such that $({H_s} - {\lambda _s}I)$ is non-singular while $({H_t} - {\lambda _t}I)$ and $({H_u} - {\lambda _u}I)$ are singular, and the vector $\left| \psi_{ts}  \right\rangle \doteq \sum\limits_{j = 2}^n {\left\langle j \right|({H_t} - {\lambda _t}I){{({H_s} - {\lambda _s}I)}^{ - 1}}\left| 1 \right\rangle } \left| j \right\rangle$ (which may be the zero vector) belongs to $S_t(\lambda_t)$ but not $S_s(\lambda_s)$. Assume the same condition for $\left| \psi_{us}  \right\rangle$. If $\left| \psi_{ts}  \right\rangle$ is not collinear with a number of $\left| \psi_{us}  \right\rangle$; $S_t(\lambda_t)$, $S_u(\lambda_u)$, and $S_s(\lambda_s)$ would be non-coincident. We prove that $\left| \psi_{ts}  \right\rangle$ and $\left| \psi_{us}  \right\rangle$ are not collinear, by contradiction. If $\left| \psi_{ts}  \right\rangle$ was collinear to a number of $\left| \psi_{us}  \right\rangle$, then for every $\alpha$ (due to subspace properties for the null-space) we have:
\begin{equation}
\left| {{\psi _{ts}}} \right\rangle  = \alpha \left| {{\psi _{us}}} \right\rangle  \Leftrightarrow \sum\limits_{j = 2}^n {\left\langle j \right|\left( {({H_t} - {\lambda _t}I) - \alpha ({H_u} - {\lambda _u}I)} \right){{({H_s} - {\lambda _s}I)}^{ - 1}}\left| 1 \right\rangle }  = 0.
\end{equation}
It would be necessary that ${{{({H_s} - {\lambda _s}I)}^{ - 1}}\left| 1 \right\rangle }$ be simultaneously an eigenket of ${{{({H_t} - {\lambda _t}I)}} }$ and ${{{({H_u} - {\lambda _u}I)}} }$. Thus, by the assumption in \ref{B2}, if there were no common eigenket for control Hamiltonians, the intersection ${\bigcap\limits_k {{S_k}({\lambda _k})} }$ would be at most $1$-dimensional. The case that $ {{\lambda _k} \in \sigma ({H_k})}$ for all $k$ is a special case of what has been proved. So \ref{B1} has been proved.
\\ Consider the case that there exists common eigenkets for control Hamiltonians. Therefore, the intersection subspaces $S_t(\lambda_t)\cap S_s(\lambda_s)$ and $S_u(\lambda_u)\cap S_s(\lambda_s)$ (which are at most $(n-2)$-dimensional) may coincide. If there are $s$ common eigenkets, with the proposed statement, the intersection ${\bigcap\limits_k {{S_k}({\lambda _k})} }$ may be at most $(1+s)$-dimensional which proves \ref{B2}.
\end{proof}
The previous theorem revealed that for each set of $\left\{ {{\lambda _k} \in \mathbb{R},k = 1,...,n - 1} \right\}$, in the case that the control Hamiltonians do not have common eigenkets, the invariant set includes at most one quantum equivalence class. This result will help to provide further useful conditions for asymptotic stochastic stability. In the rest of this paper, assume that the control Hamiltonians do not share any eigenkets, which is not very restrictive.
\par Now consider the case that $\left| \psi  \right\rangle \in A$: Knowing that $\left\langle {\psi } \mathrel{\left | {\vphantom {\psi  {{\psi _f}}}} \right. \kern-\nulldelimiterspace} {{{\psi _f}}} \right\rangle  = 0$, let us study the invariance for this situation. For an infinitesimal time duration, the inner product would evolve as follows:
\begin{equation}
\operatorname{E}\left[ {\left\langle {{{\psi _f}}}
 \mathrel{\left | {\vphantom {{{\psi _f}} {\psi (dt)}}}
 \right. \kern-\nulldelimiterspace}
 {{\psi (dt)}} \right\rangle } \right] = \frac{{ - i}}{\hbar }{\sum\limits_k {{\alpha _k}\left( {\operatorname{Im} \left( {\left\langle {{{\psi _f}}}
 \mathrel{\left | {\vphantom {{{\psi _f}} {{H_k}\left| \psi  \right.}}}
 \right. \kern-\nulldelimiterspace}
 {{{H_k}\left| \psi  \right.}} \right\rangle } \right)} \right)} ^2}\left\langle {{{\psi _f}}}
 \mathrel{\left | {\vphantom {{{\psi _f}} {{H_k}\left| \psi  \right.}}}
 \right. \kern-\nulldelimiterspace}
 {{{H_k}\left| \psi  \right.}} \right\rangle dt.
\label{380}
\end{equation}
Thus, if $\operatorname{Im} \left( {\left\langle {{{\psi _f}}}
 \mathrel{\left | {\vphantom {{{\psi _f}} {{H_k}\left| \psi  \right.}}}
 \right. \kern-\nulldelimiterspace}
 {{{H_k}\left| \psi  \right.}} \right\rangle } \right) \ne 0$ for at least one $k$, the quantum trajectory is expected to escape the orthogonal subspace $ {\left[ {\left| {{\psi _f}} \right\rangle } \right]^ \bot }$. On the other hand, assume that there exists $\left| \psi  \right\rangle  \in {\left[ {\left| {{\psi _f}} \right\rangle } \right]^ \bot }$ such that for at least one $k$, $\left\langle {{{\psi _f}}} \mathrel{\left | {\vphantom {{{\psi _f}} {{H_k}\left| \psi  \right.}}} \right. \kern-\nulldelimiterspace} {{{H_k}\left| \psi  \right.}} \right\rangle  = r{e^{i\theta }} \ne 0$. Putting $\left| {\hat \psi } \right\rangle  = {e^{ - i\theta }}\left| \psi  \right\rangle $ (which also belongs to ${\left[ {\left| {{\psi _f}} \right\rangle } \right]^ \bot }$) gives $\operatorname{Im} \left( {\left\langle {{{\psi _f}}}
 \mathrel{\left | {\vphantom {{{\psi _f}} {{H_k}\left| {\hat \psi } \right.}}} \right. \kern-\nulldelimiterspace} {{{H_k}\left| {\hat \psi } \right.}} \right\rangle } \right) = 0$. Therefore, the problem in this situation reduces to finding the minimal set of control Hamiltonians $\left\{ {{H_k}} \right\}$ such that:
\begin{equation}
\label{390}
\left\{ {\left| \psi  \right\rangle  \in {{\left[ {\left| {{\psi _f}} \right\rangle } \right]}^ \bot }\left| {\exists k:\left\langle {{{\psi _f}}}
 \mathrel{\left | {\vphantom {{{\psi _f}} {{H_k}\left| \psi  \right.}}}
 \right. \kern-\nulldelimiterspace}
 {{{H_k}\left| \psi  \right.}} \right\rangle  = 0} \right.} \right\} = \emptyset.
\end{equation}
Now the following theorem can be stated:
\begin{theorem}
\label{TH40}
Consider the SSE (\ref{SSE}). Assume that A\ref{A10} to A\ref{A50} hold. Then, the quantum trajectories starting from ${\left[ {\left| {{\psi _f}} \right\rangle } \right]}^ \bot $, will escape it with probability $1$, i.e., the set $A$ in Theorem \ref{TH30} is not an invariant set.
\end{theorem}
\begin{proof}
Based on the statement above, it suffices to show that (\ref{390}) holds. For every ${\left| \psi  \right\rangle  \in {{\left[ {\left| {{\psi _f}} \right\rangle } \right]}^ \bot }}$, one may write $\left| \psi  \right\rangle  = \sum\limits_{j = 2}^n {{c_j}\left| j \right\rangle }$. Also, in this coordinate, each of the control Hamiltonians can be written as ${H_k} = \sum\limits_{h = 1}^n {\sum\limits_{l = 1}^n {{c_{khl}}\left| h \right\rangle } } \left\langle l \right|$ (of course with some restrictions on $c_{khl}$). Thus, we have:
\begin{equation}
\left\langle {{\psi _f}} \right|{H_k}\left| \psi  \right\rangle  = \sum\limits_{j = 2}^n {{c_{k1j}}{c_j}\left| j \right\rangle }. 
\end{equation}
Also (\ref{390}) holds if the system of linear equations
\begin{equation}
\left( {\begin{array}{*{20}{c}}
  {{c_{112}}}\hspace{0.5cm}{{c_{113}}} \hspace{0.25cm}\cdots \hspace{0.25cm}{{c_{11n}}} \\ 
   \vdots \hspace{0.9cm} \vdots \hspace{0.4cm}  \ddots \hspace{0.5cm} \vdots  \\ 
  {{c_{n - 112}}}\hspace{0.3cm}{{c_{n - 113}}} \cdots  {{c_{n - 11n}}} 
\end{array} } \right)\left( {\begin{array}{*{20}{c}}
  {{c_2}} \\ 
  {{c_3}} \\ 
   \vdots  \\ 
  {{c_n}} 
\end{array}} \right) = {\mathbf{0}}
\end{equation}
does not have a nontrivial solution. But this condition is always provided due to linear independence of $\left \{ H_k\right \}$ in A\ref{A50}.
\end{proof}
\begin{remark}
Based on the proof of Theorem \ref{TH40}, $m=n-1$ is the minimal number of independent control Hamiltonians, which is stated in A\ref{A50}.
\end{remark}
Now, based on these two stated theorems, one of the striking features of this theory can be stated. By Theorem \ref{TH30}, it is revealed that the right invariant set of quantum states, is at most $1$-dimensional for each choice of $\left\{\lambda_k\right\}$. On the other hand, Theorem \ref{TH40} revealed that set $A$ in Theorem \ref{TH30} is not right invariant. Now let us investigate set $B$. Consider that the quantum system is initiated in the quantum equivalence class $\left[ {\left| {{\psi _0}} \right\rangle } \right]$ almost surely and there exists a set $\left\{ {{\lambda _k} \in \mathbb{R} } \right\}$ such that for all $k$, $\langle {\psi _f}|{H_k} - {\lambda _k}I\left| \psi  \right\rangle  = 0$. In order to investigate the right invariance property, one must inspect whether or not the dynamics inspired by (\ref{SSE}) preserve the vanishing $LV\left( {\left| {{\psi ^{\left| {{\psi _0}} \right\rangle }}\left( t \right)} \right\rangle } \right)$. To this end, the following theorem shows that the invariant set is exclusively containing ${{\left[ {\left| {{\psi _f}} \right\rangle } \right]} }$.
\begin{theorem}
Assume that A\ref{A10} to A\ref{A50} hold. If the control Hamiltonians do not share any common eigenkets, then the invariant set of (\ref{SSE}) exclusively includes ${{\left[ {\left| {{\psi _f}} \right\rangle } \right]} }$.
\end{theorem}
\begin{proof}
 Let us consider $LV\left( {\left| {{\psi ^{\left| {{\psi _0}} \right\rangle }}\left( dt \right)} \right\rangle } \right)$. By SSE (\ref{SSE}), one should find a set $\left\{ {{\hat \lambda _k} \in \mathbb{R} } \right\}$ such that: $$\left\langle {{\psi _f}} \right|\left( {{H_K} - {{\hat \lambda }_k}I} \right)\left( {I + \left( {\frac{{ - i}}{\hbar }{H_0} - k{{\left( {X - \left\langle X \right\rangle } \right)}^2}} \right)dt + \sqrt {2k} \left( {X - \left\langle X \right\rangle } \right)dW} \right)\left| {{\psi _0}} \right\rangle  = 0.$$
Note that $u_k(0)=0$ and thus the effect of $H_k$ vanishes. The presence of Wiener process implies that both of the following equalities must simultaneously hold:$$\left\langle {{\psi _f}} \right|\left( {{H_K} - {{\hat \lambda }_k}I} \right)\left( {I + \left( { - k{{\left( {X - \left\langle X \right\rangle } \right)}^2}} \right)dt + \sqrt {2k} \left( {X - \left\langle X \right\rangle } \right)dW} \right)\left| {{\psi _0}} \right\rangle  = 0$$ and 
\begin{equation}
\left\langle {{\psi _f}} \right|\left( {{H_K} - {{\hat \lambda }_k}I} \right)\left( {I + \frac{{ - i}}{\hbar }{H_0}dt} \right)\left| {{\psi _0}} \right\rangle  = 0.
\label{420}
\end{equation}
Also, it is intuitively obvious that $\lambda_k$ is uniformly continuous in $t$ and if written as $\hat\lambda_k=\lambda_k+\delta_k$ for real $\delta_k$, then $\delta_k \rightarrow0$  as $dt\rightarrow0$. Let us investigate (\ref{420}). Terms can be reordered to obtain: $$\left\langle {{\psi _f}} \right|\left( {{H_K} - {{ \lambda }_k}I} \right)\left( {\frac{{ - i}}{\hbar }{H_0}dt} \right)\left| {{\psi _0}} \right\rangle  = {\delta _k}\left\langle {{{\psi _f}}}
 \mathrel{\left | {\vphantom {{{\psi _f}} {{\psi _0}}}}
 \right. \kern-\nulldelimiterspace}
 {{{\psi _0}}} \right\rangle \left( {1 + \frac{{ - idt}}{\hbar }{\lambda _{HF}}} \right).$$ The second term in the RHS represents a second order perturbation, which is negligible as $dt\rightarrow0$:
\begin{equation}
\left\langle {{\psi _f}} \right|\left( {{H_K} - {{ \lambda }_k}I} \right)\left( {\frac{{ - i}}{\hbar }{H_0}dt} \right)\left| {{\psi _0}} \right\rangle  = {\delta _k}\left\langle {{{\psi _f}}}
 \mathrel{\left | {\vphantom {{{\psi _f}} {{\psi _0}}}}
 \right. \kern-\nulldelimiterspace}
 {{{\psi _0}}} \right\rangle .
\label{430}
\end{equation}
 Thus the question reduces to: If there is a set $\left\{\delta_k\in \mathbb{R}\right\}$, such that for the (at most) $1$-dimensional members of $B$ in Theorem \ref{TH40}, (\ref{430}) holds for all $k$?
\\ Define $\hat{H}_k \doteq \left( {{H_K} - {{ \lambda }_k}I} \right)\left( {\frac{{ - i}}{\hbar }{H_0}dt} \right)$. One should try to find the set $\left\{\delta_k\in \mathbb{R}\right\}$ such that $\left\langle {{\psi _f}} \right|{{\hat H}_K} - {\delta _k}I\left| {{\psi _0}} \right\rangle  = 0$ and $\left\langle {{{\psi _f}}}
 \mathrel{\left | {\vphantom {{{\psi _f}} {{H_k} - {\lambda _k}I\left| \psi  \right.}}}
 \right. \kern-\nulldelimiterspace}
 {{{H_k} - {\lambda _k}I\left| \psi_0  \right.}} \right\rangle  = 0$, for all $k$. This is similar to what was tried in the proof of  Theorem \ref{TH40} with the difference that in this case, there are $2(n-1)$ solution spaces one of which is at most $n-1$-dimensional. The assumption of not sharing any eigenkets for $H_k$'s implies that the common solution of $\left\langle {{{\psi _f}}} \mathrel{\left | {\vphantom {{{\psi _f}} {{H_k} - {\lambda _k}I\left| \psi  \right.}}} \right. \kern-\nulldelimiterspace} {{{H_k} - {\lambda _k}I\left| \psi_0  \right.}} \right\rangle  = 0$ includes at most one independent ket. In order to keep the same solution to be the solution for all of $\left\langle {{\psi _f}} \right|{{\hat H}_K} - {\delta _k}I\left| {{\psi _0}} \right\rangle  = 0$, (or in other words, two functionals $\left\langle {{\psi _f}} \right|({H_K} - {\lambda _k}I)$ and $\left\langle {{\psi _f}} \right|(\frac{{ - idt}}{\hbar }\left( {{H_K} - {\lambda _k}I} \right){H_0} - {\delta _k}I)$ share the same kernel), by Lemma \ref{L5}, there are only two possibilities: 
\begin{enumerate}
\item{ Whether $H_0=cI$ for some scalar complex $c$, which is impossible,}
\item{ $\left| {{\psi _0}} \right\rangle$ is an eigenket for $H_0$ which implies $\delta_k=0$.}
\end{enumerate}
 Regarding this explanation, the only possibilities to be included in the invariant set are the eigenkets of $H_k$ making $LV\left( {\left| {{\psi }} \right\rangle }\right)=0$. On the other hand, even if $H_0$ is degenerate (the stationary states can be a super-position of eigenstates with the same enery level), all of the stationary states apart from ${{\left[ {\left| {{\psi _f}} \right\rangle } \right]}}$  have to include in ${{\left[ {\left| {{\psi _f}} \right\rangle } \right]}^ \bot }$. However, by Theorem \ref{TH40}, ${{\left[ {\left| {{\psi _f}} \right\rangle } \right]}^ \bot }$ is not invariant. By the virtue that ${\left| {{\psi _f}} \right\rangle  }={\left| {{1}} \right\rangle}$ is degenerate, no stationary states can be in the superposition of the eigenspace corresponding to ${\left| {{1}} \right\rangle}$ and the eigenspaces in ${{\left[ {\left| {{\psi _f}} \right\rangle } \right]}^ \bot }$. In light of the facts outlined above and by the use of \ref{B3},  the right invariant set is solely restricted to ${{\left[ {\left| {{\psi _f}} \right\rangle } \right]}}$.
\end{proof}
\begin{remark}
It is worth noting that this proof implies that \textit{even if $H_0$ is degenerate in the eigenspaces except for ${\left| {{1}} \right\rangle}$ }, if A\ref{A10} to A\ref{A50} hold and $\left\{H_k\right\}$ do not share any eigenkets, the $\Omega$-limit set merely includes ${{\left[ {\left| {{\psi _f}} \right\rangle } \right]}}$. Although the non-degeneracy condition for $H_0$ was essential in almost all of the reported works on Lyapunov control of Schr\"odinger equation, in this paper, it was waived by the virtue of the proposed theory.
\end{remark}

\section{Computer experiment}\label{p14}
In order to qualify and illustrate the proposed theory, it is applied to a $2-$level quantum system. Although these systems are among the simplest quantum systems, they are of variety of applications in quantum computing techniques. Simply assume that the Hamiltonian is of the following form:
\[H = {\sigma _z} + u_1(t){\sigma _y}.\]
This form may model a Fermion in an orthogonal electromagnetic field $B = {B_z}\hat z + {B_y}(t)\hat y$, the state of an artificial atom in a superconducting qubit or many other $2-$level quantum systems.
\\ Also, put $X = {\sigma _z}$, assign the measurement strength $k = 0.1$ and ${\alpha _1} = 5$. Our aim is to manipulate and switch this quantum system between $\left|  \uparrow  \right\rangle$ and $\left|  \downarrow  \right\rangle$ which are the eigenstates of $\sigma _z$, while we are continuously measuring this observable. These conditions obey A\ref{A10} to A\ref{A50} and the proposed theory suggests that the final desired states are asymptotically stable. The simple paths of a quantum rajectoriy are of the form $\left| {\psi (t)} \right\rangle  = \left( {\begin{array}{*{20}{c}}
  {{c_1}(t)} \\ 
  {{c_2}(t)} 
\end{array}} \right)$ in the Pauli notation.
\par Figure \ref{c1c2} illustrates a simple path of the quantum trajectory driven by the proposed method. The system is switched from $\left|  \downarrow  \right\rangle  \triangleq \left( {\begin{array}{*{20}{c}}
  0 \\ 
  1 
\end{array}} \right)$ to $\left|  \uparrow  \right\rangle  \triangleq \left( {\begin{array}{*{20}{c}}
  1 \\ 
  0 
\end{array}} \right)$.  Figure \ref{contrl} shows the control signal $u_1(t)$. Figure \ref{Lyapunov} shows the Lyapunov value for the sample path. The expected value of the observable $X=\sigma_z$ is shown in Figure \ref{Xexpec}. Also, the simple path generated by the SSE and the proposed manipulation algorithm is shown on the Bloch sphere in Figure \ref{blochsphere}.
\section{Conclusion and further research} \label{p15}
Continuous measurement is certainly a groundbreaking point to feedback control of quantum systems. Mesoscopic quantum systems are competitively founding their way to quantum computing applications. Among the main influencing aspects of these systems to make them implementable, are their capability to get written, controlled and read-out easily and fast due to their short coherence time. This paper, considers and takes into account all of these three aspects. As a result, homodyning, as a promising way to continuous measurement is getting attention for qubit manipulation purposes. Thus, the need for a capable stabilization algorithm is inevitable in this area. Also, these systems are easily modelled and driven by the SSE when their dynamics is unravelled.
\par Our goal was to propose a stabilization algorithm for quantum systems when they are continuously measured. Up to some conditions on the measurement observable and the Hamiltonian of the system, this goal was achieved. Fortunately, the conditions are not restrictive and are satisfied in most of experimental setups; for instance, as shown in the computer experiment section, an stochastic $2-$level quantum system can be stabilized with a single control manipulator. Also, this algorithm, likewise other Lyapunov-based algorithms, is robust to small dynamical perturbations and thus, the control history can be used in off-line manner. On of the main advantages of this theory is that it works for degenerate Hamiltonian. Despite some existing algorithms in the literature that stabilize deterministic Schr\"odinger equation, which require the Hamiltonian to be $\lambda-$degenerate (which is more restrictive than degeneracy condition), this theory does not require degeneracy of the Hamiltonian.
\par Further research will focus on extending the proposed theory to output feedback scheme. Also, another potential area would be the applications of this theory to quantum computing frameworks.

\begin{figure}
\centering
\includegraphics[width=0.8\textwidth]{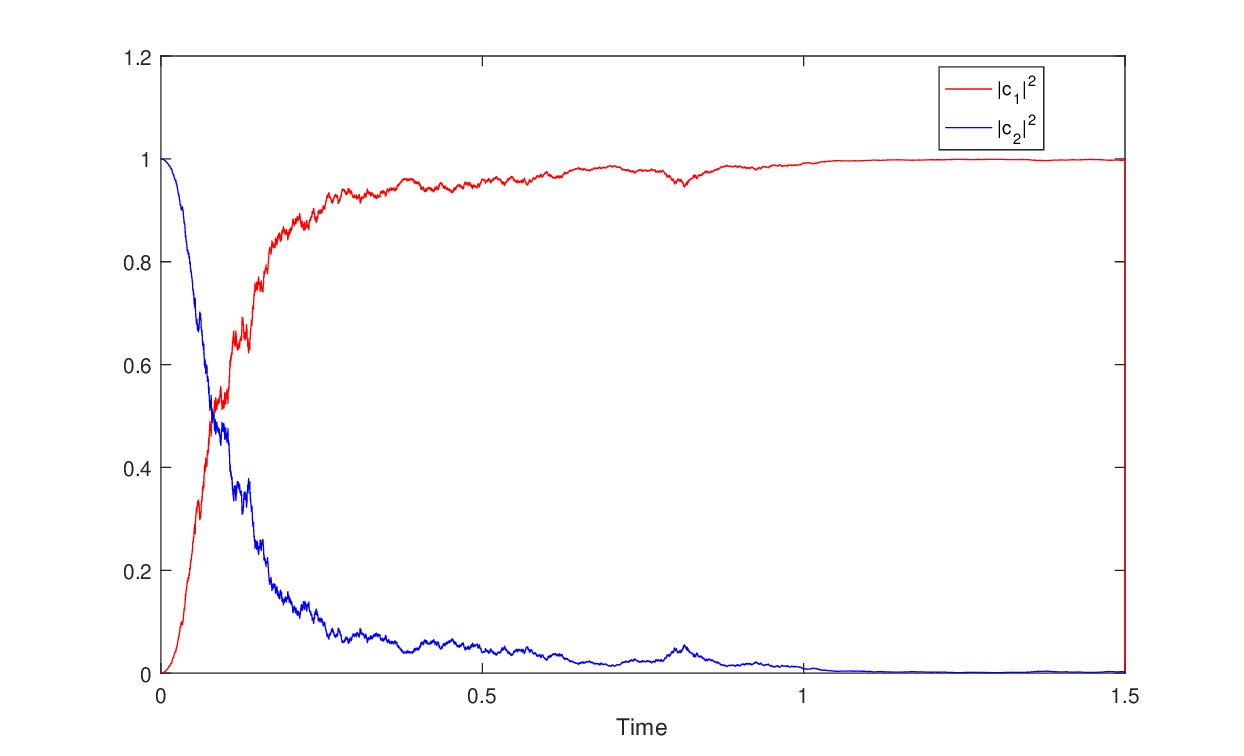}
\caption{Transition probabilities. The system rotates from $\left|  \downarrow  \right\rangle$ to $\left|  \uparrow  \right\rangle$. As it is shown, the proposed control algorithm drives the stochastic system, asymptotically from the initial state to the desired final state. This figure illustrates how this algorithm is useful in order to design controlled quantum gates. This is an example of simple CNOT gate.}
\label{c1c2}
\end{figure}

\begin{figure}
\centering
\includegraphics[width=0.8\textwidth]{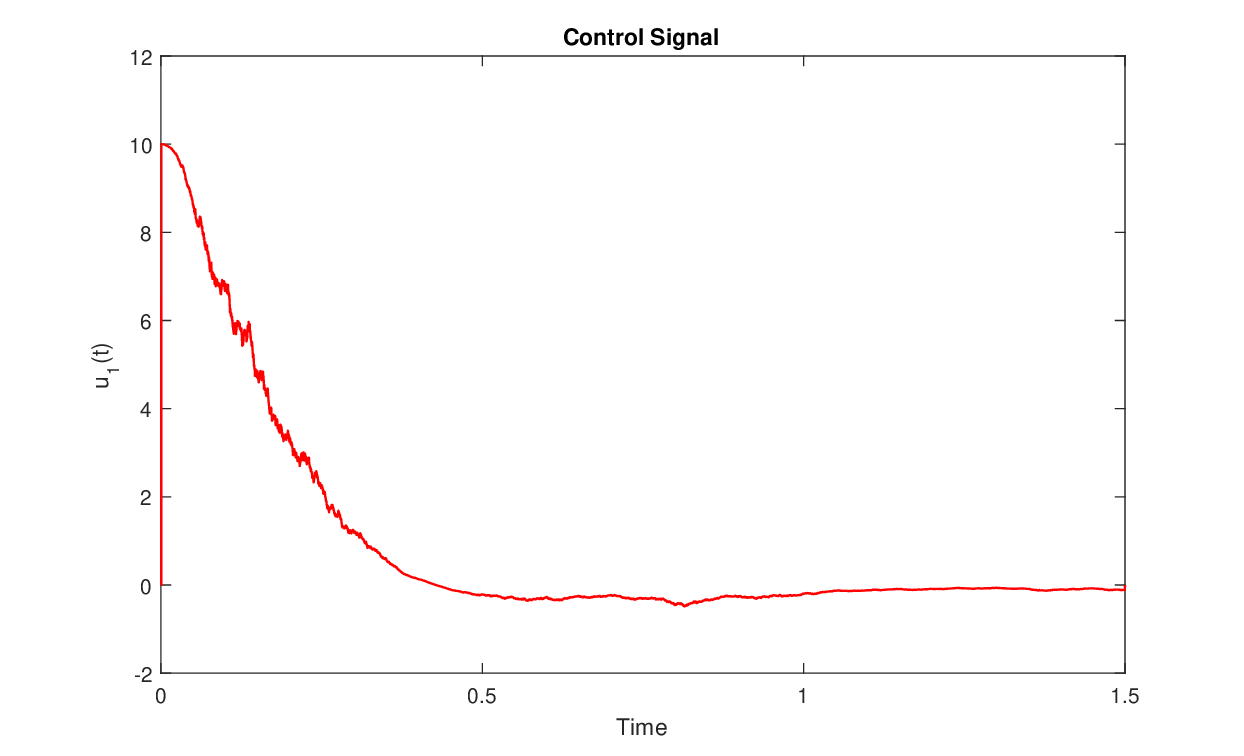}
\caption{Control signal $u_1(t)$. At the beginning, the control effort is high in order to decrease the Lyapunov value. Also, the effect of quantum jumps are substantial in the meanwhile, this is because at the beginning and the end of the transition, the quantum states are eigenkets of $X$.}
\label{contrl}
\end{figure}

\begin{figure}
\centering
\includegraphics[width=0.8\textwidth]{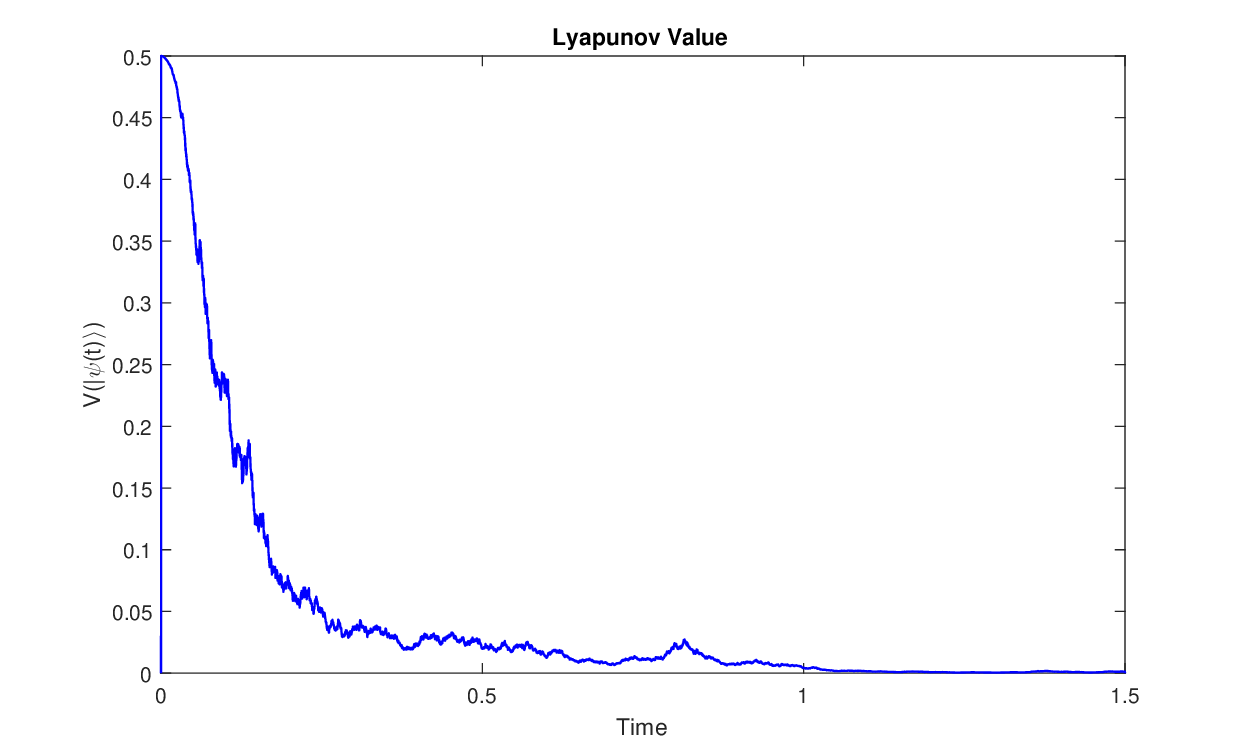}
\caption{Lyapunov value $V(\left| \psi (t) \right\rangle )$. The Lyapunov value is increasing stochastically as predicted. Also, the asymptotic stability ensures no invariant set in the meanwhile.}
\label{Lyapunov}
\end{figure}

\begin{figure}
\centering
\includegraphics[width=0.8\textwidth]{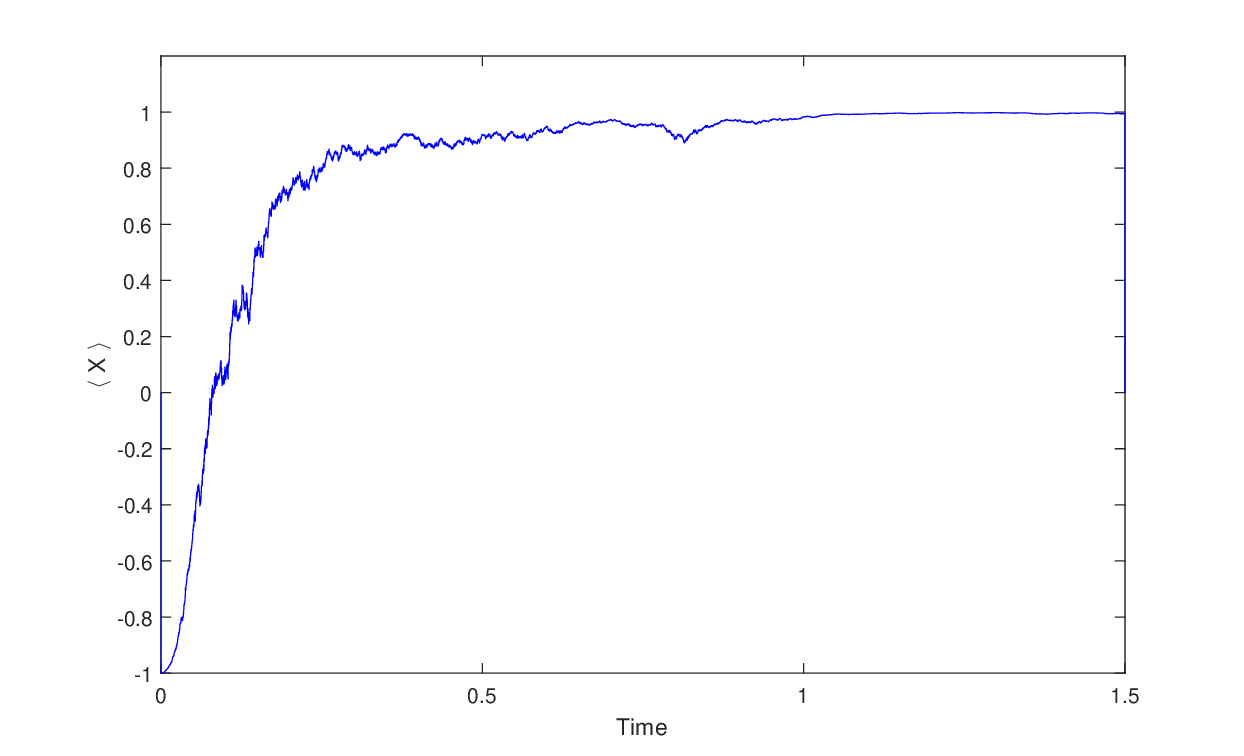}
\caption{Expectation value $\left\langle X \right\rangle$. This is also the measurement record and shows the extracted information about the state of the system.}
\label{Xexpec}
\end{figure}

\begin{figure}
\centering
\includegraphics[width=1\textwidth]{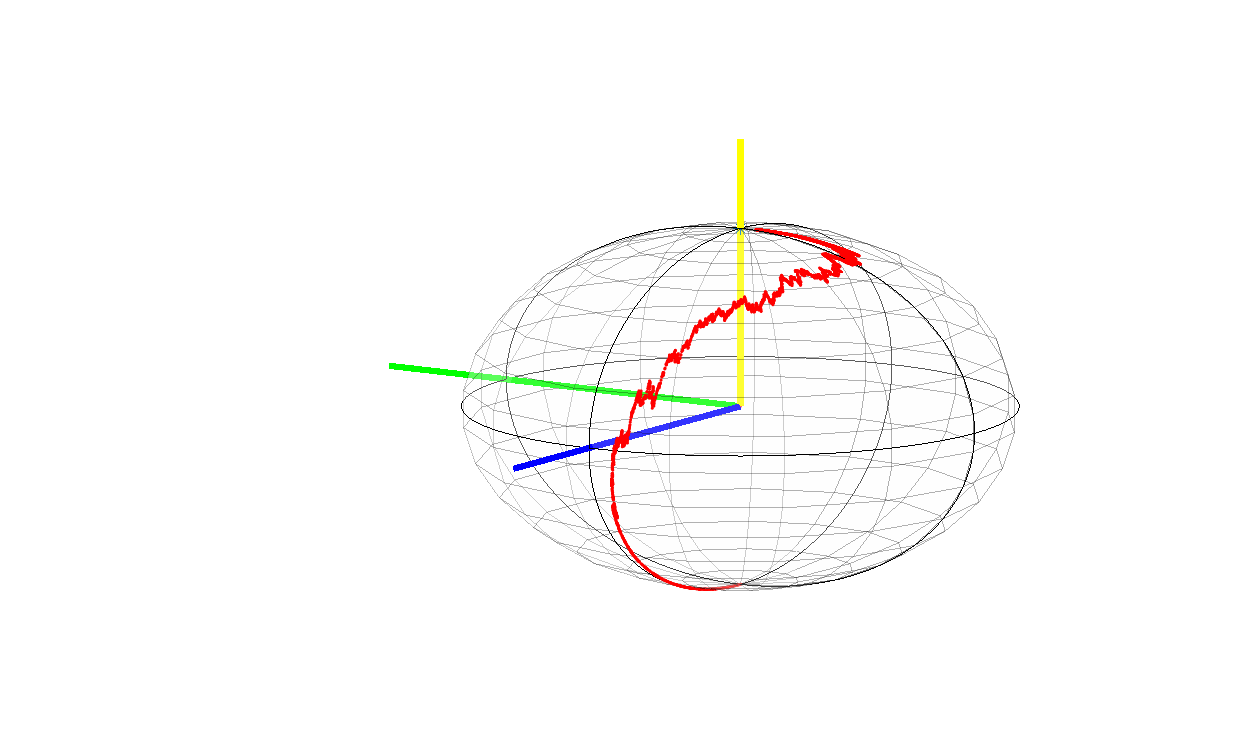}
\caption{A simple path of state tractores on the Bloch sphere. Yellow, blue and green axes are $Z+$, $Y+$ and $X+$ respectively. The simple path transfers from $\left|  \downarrow  \right\rangle$ to $\left|  \uparrow  \right\rangle$. Quantum jumps in the meanwhile are illustrated.}
\label{blochsphere}
\end{figure}

\bibliographystyle{plainnat}
\bibliography{bib1}



\end{document}